\documentclass[11pt]{article}
\usepackage[left=1in,right=1in,top=1in,bottom=1in]{geometry}

\usepackage[utf8]{inputenc}
\usepackage{xcolor,amsmath,amsfonts}
\usepackage{latexsym,amssymb,amsbsy, amsthm}
\usepackage{bbm, enumitem}
\usepackage{graphicx}
\usepackage{algorithm,algpseudocode}
\usepackage{float}
\usepackage[font=small]{caption}
\usepackage[normalem]{ulem}
\usepackage{wrapfig}
\usepackage{authblk}
\usepackage{tabularx}
\usepackage{makecell}
\usepackage{multirow}
\usepackage{grffile}
\usepackage{caption}
\usepackage{subcaption}

\newcommand \sg[1] {{\color{black}#1}}

\definecolor{blue}{rgb}{0,0,0}

\newcommand{\etal}{{\em et~al.~}}
\newcommand{\abs}[1]{\lvert #1 \rvert}
\newcommand{\E}{\mathcal{E}}			
\newcommand{\bigO}{\mathcal{O}}			
\newcommand{\minmin}{\textsc{Min-Min}}			
\newcommand{\Reff}{{R}_{\text{eff}}}
\newcommand{\Reffk}[1]{{R}^{(#1)}_{\text{eff}}}	

\usepackage{enumerate}
\DeclareMathOperator*{\argmin}{arg\,min}

\allowdisplaybreaks 
\newtheorem{theorem}{Theorem}

\newtheorem{lemma}{Lemma}
\newtheorem{corollary}{Corollary}
\newtheorem{definition}{Definition}
\usepackage[pdftex]{hyperref}

\algnewcommand\INPUT{\item[{\textbf{Input:}}]}
\algnewcommand\RETURN{\item[{\textbf{Return:}}]}

\title{\vspace{-40pt} Electrical Flows over Spanning Trees \vspace{15 pt}}

\author[1]{Swati Gupta}
\author[2] {Ali Khodabakhsh}
\author[1]{Hassan Mortagy}
\author[2] {Evdokia Nikolova}

\affil[1]{\small Georgia Institute of Technology, Atlanta GA 30332, USA \protect \\
{\small \tt \{swatig,hmortagy\}@gatech.edu}}
\affil[2]{\small University of Texas at Austin, Austin TX 78712, USA \protect \\
{ \tt {ali.kh@utexas.edu} \protect \\
{nikolova@austin.utexas.edu}}}
\date{}

\begin{document}
\maketitle
\vspace{-30pt}
\begin{abstract}
The {\it network reconfiguration} problem seeks to find a rooted tree $T$ such that the energy of the (unique) feasible electrical flow over $T$ is minimized. The tree requirement on the support of the flow is motivated by operational constraints in electricity distribution networks. The bulk of existing results on convex optimization over vertices of polytopes and on the structure of electrical flows do not easily give guarantees for this problem, while many heuristic methods have been developed in the power systems community as early as 1989. Our main contribution is to give the first provable approximation guarantees for the network reconfiguration problem. We provide novel lower bounds and corresponding approximation factors for various settings ranging from $\min\{\bigO(m-n), \bigO(n)\}$ for general graphs, to $\bigO(\sqrt{n})$ over grids with uniform resistances on edges, and $\bigO(1)$ for grids with uniform edge resistances and demands. To obtain the result for general graphs, we propose a new method for (approximate) spectral graph sparsification, which may be of independent interest. \sg{Using insights from our theoretical results}, we propose a general heuristic for the network reconfiguration problem that is orders of magnitude faster than existing methods in the literature, while obtaining comparable performance.
\vspace{5pt}
\end{abstract}


%


\section{Introduction}\label{sec:intro}
Electricity distribution and transmission networks have been a rich source of non-convex problems with combinatorial structure that have helped discover limitations of, as well as advance, the theory of discrete and continuous optimization \cite{coffrin2015qc,gan2014exact,molzahn2019survey}. In this paper, we consider the {\it network reconfiguration problem} which seeks to find a rooted tree such that the energy of the electrical flow over the tree is minimized. This problem is motivated by the operational requirements of electricity distribution networks. Distribution networks contain low-voltage power lines connecting substations to end consumers.  They are typically built as mesh networks, i.e., containing cycles, but operated as radial networks, i.e., power is sent along a tree rooted at a substation, with leaves corresponding to end consumers. A real-world example of such radial configuration is shown in Figure~\ref{fig:realnetwork}.
The tree structure is achieved 
by turning switches on or off so that there are no cycles \sg{(i.e., by adding or deleting edges to configure a specific tree)}. This tree structure is desirable for the security of an electricity distribution network: a fault can be more easily isolated in practice when the downstream edge from any fault gets disconnected from the root, while cycles in the operational network might result in compromising large parts of the network. Distribution network reconfiguration is a key tool used by operators to balance the electric load across power lines and mitigate power losses.\footnote{{Throughout the paper we use energy and loss interchangeably.}} Minimizing power loss becomes a major concern in distribution networks as they operate with low voltage power lines, and hence can admit fairly big energy losses \cite{sarfi1994survey} 
(around 10\%, in contrast with the much lower losses at the high-voltage transmission level of about 1-3\%).
We give a formal description of the network reconfiguration problem next.

\subsection{Problem Formulation} {Let $G=(V,E)$ be an undirected graph ($\abs{V}=n$, $\abs{E}=m$), with root $r\in V$}, resistances $r_e>0$ for each edge\footnote{Edges are referred to as {\it lines}, root is the location of the {\it substation}, and demands are often referred to as {\it loads} in the power systems community.} $e\in E$ and demands $d_i\geq 0$ for each node $i\in V\backslash \{r\}$ supplied by the root node (thus $d_r = - \sum_{i\in V\setminus \{r\}} d_i$). {Let $\delta^+(v)$ and $\delta^-(v)$ denote the sets of incoming and outgoing edges of $v$, after fixing an arbitrary orientation on the edges}. Then, the {\em network reconfiguration} problem is to minimize the energy of the feasible flow {(also, referred to as the power loss)} such that the support of the flow  is acyclic:
\begin{align}
    \tag{P0} \min \quad& \E(f) := \sum_{e \in E} r_e f_e^2\label{eq:p0}\\
    \text{subject to} \quad&\sum_{e \in \delta^+(i)} f_e - \sum_{e\in \delta^-(i)} f_e = d_i, ~\forall i\in V,\label{flow}\\
    & \text{support}(f) \text{ is acyclic} \label{support}, 
\end{align}
where $f$ is any feasible flow satisfying the demands \eqref{flow}\footnote{Here we use a simplified linear flow model, similar to \cite{Madry_2010}. In reality, power flow equations are nonlinear and result in non-convex optimization problems  \cite{low2014convex,madani2014convex,molzahn2014investigation,peng2016distributed}. We refer the reader to a recent survey on relaxations and approximations of power flow equations \cite{molzahn2019survey}. In contrast, here we aim to relax the non-linearity of the power flow model, and instead focus on the combinatorial aspect of the optimization problem in \eqref{eq:p0}.}, the support of $f$ is constrained to be acyclic (and therefore, a tree rooted at $r$) \eqref{support}, and the objective \eqref{eq:p0} is to minimize the energy of the flow. While there may be more than one feasible flow satisfying the demands in general, an electrical flow minimizes the energy subject to meeting the demands. 
Moreover, {given an $r$-rooted tree}, there is a unique flow $f$ on the tree that satisfies the demands: $f_e = \sum_{i \in \text{succ}(e)} d_i$, where $\text{succ}(e)$ is the set of nodes that connect to the root through $e$. {Also, note that any $r$-rooted tree can be augmented to be a spanning tree in the graph at no additional cost.\footnote{Contract the support of the flow and add edges with 0 flow to construct a spanning tree.}} Therefore, the network reconfiguration problem \eqref{eq:p0} {is equivalent to}:\footnote{We assume that the demands are real numbers, while in energy systems, demands are usually complex numbers $d=p+\mathbf{i}q$, capturing the active $(p)$ and reactive $(q)$ parts of the demand. In this case, the objective function can be decomposed into two additive parts, in which one is only a function of real demands ($p$), and the other is only a function of the reactive part ($q$). All our approximation guarantees hold for this more realistic objective as well, because the proposed solutions would guarantee the same approximation factor for both parts of the objective.}
\begin{equation*}
\min_{T\in \mathcal{T}} \quad \sum_{e\in T} r_e\left( \sum_{i\in \text{succ}(e)} d_i\right)^2, 
\tag{P1}
\label{eq:p1}
\end{equation*}
where $\mathcal{T}$ is the set of all spanning trees of $G$.

Although distribution grid practitioners and power systems researchers have used a broad range of heuristics to get reasonable solutions in practice (see e.g., \cite{civanlar1988distribution,baran1989network,kumar2014power}), little is known about provable bounds for the network reconfiguration problem. 
Khodabakhsh \etal \cite{khodabakhsh2018submodular} provided an integer programming formulation (with a cubic objective) for \eqref{eq:p1} involving $\bigO(n^3)$ constraints and variables using Martin's extended formulation for spanning trees \cite{martin1991}. Note that the network reconfiguration problem requires minimizing a quadratic function over the vertices of the general flow polytope (Theorem 7.4 in \cite{bertsimas_97})\footnote{{\sg{The support of the vertices of the general flow polytope is a tree}}. We use the term `polytope' instead of `polyhedron' because we assume $r_e > 0$ for all edges $e \in E$.}. In particular, if all the resistances are uniform, then the problem is equivalent to finding a vertex with the smallest Euclidean norm, which is known to be NP-hard (see for example Lemma 4.1.4 in \cite{haddock_2018}). We next present an overview of the related work on this problem.

\subsection{Related Work} Khodabakhsh~\etal \cite{khodabakhsh2018submodular} showed that \eqref{eq:p1} is NP-hard even for the \emph{uniform} case where all the resistances and demands are equal to one, i.e., $r_e=1$ for all $e\in E$ and $d_i=1$ for all $i\in V\backslash \{r\}$. They showed that \eqref{eq:p1} can be cast as a supermodular minimization problem subject to a matroid base constraint. Inspired by submodular maximization, a local search algorithm was also proposed in \cite{khodabakhsh2018submodular}. Their local search algorithm is equivalent to a well-known heuristic called \emph{branch exchange}, which moves locally by swapping edges starting with a spanning tree. The best known guarantees for this local search come from the generic result on constrained submodular maximization \cite{lee2010maximizing}, however the approximation factor depends on an upper bound of the objective function (and thus can be arbitrarily bad). A brute-force method was proposed by \sg{Morton and Mareels} \cite{morton2000efficient}, via enumerating all spanning trees and calculating the losses by adjusting the losses of the previous spanning tree. 
Khodr and Martinez \cite{khodr2009integral} used Benders decomposition to study the joint problem of network reconfiguration and optimal power flow. Using this decomposition, they decomposed the problem into master and slave subproblems, where they used a mixed-integer non-linear solver for the master problem. 
These methods are computationally intractable for large networks. For general polytopes, the \sg{best approximation} one can hope to get \sg{in polynomial time}, for minimizing \sg{a general} strongly convex function over integral points in a polytope is $O(n^2-n)$ where $n$ is the dimension of the polytope \cite{Baes_2012}. Moreover, for minimizing quadratic functions (like our objective), Hildebrand et al. \cite{Hildebrand_2016} give an FPTAS for rounding to an integer point in the flow polytope; see Appendix \ref{sec:convex} for more details. However, 
\sg{different techniques are required for the network reconfiguration problem, since we need to minimize the quadratic function over {\it vertices} of the flow polytope, and not over {\it feasible integral} points.}

Besides distribution networks, switching problems have also been studied in electricity transmission networks, such as optimal transmission switching \cite{fisher2008optimal,fattahi2018bound,kocuk2017new} and maximum transmission switching \cite{grastien2018maximum}. Despite the hardness of these problems \cite{kocuk2016cycle,lehmann2014complexity}, Grastien~\etal \cite{grastien2018maximum} show how to achieve a 2-approximation for maximum transmission switching on cacti graphs. However, for transmission networks there is no requirement on the support to be acyclic, 
thus making our problem structurally very different from those above.

In {the} network reconfiguration problem, relaxing the tree constraint \eqref{support} results in the well-studied problem of computing electrical flows as they uniquely minimize energy \cite{peng_2011,kelner_2013,cohen_2019}. However, existing results in spectral sparsification \cite{spielman_2008,spielman_spars}, \sg{that sparsify a graph without changing the energy much}, do not extend to our setting since they change the resistances on the remaining edges to compensate for edge deletions. Many existing heuristics involve iterative edge-deletion ({e.g., \cite{shirmohammadi1989reconfiguration}}) using the electrical flow values in the resultant graph, but offer no provable guarantees. It is also known that the total stretch of a tree can bound the ratio between the energy of that tree with the original graph \cite{kelner_2013}; however, using low-stretch trees would only provide an $\tilde{\bigO}(m)$ approximation \cite{ofer_2012,karp_1995}. We give a more detailed review of related work in Appendix~\ref{sec:relwork} and a summary of contributions next. 

\vspace{-1em}
\subsection{Summary of Contributions}
Ours is the first paper to provide a provable approximation guarantee for the network reconfiguration problem, to the best of our knowledge. We delve deeper into the problem structure to construct novel lower bounds and give new ways of graph sparsification while maintaining original edge-resistances. Our theoretical insights also lead to a significant improvement in computations. We summarize our contributions below: 

\begin{itemize}[leftmargin=*]
    \item[(a)] {\it Flow Relaxation and new {\sc RIDe} algorithm:} We relax the spanning tree constraint, reducing the problem to finding a minimum energy electrical flow; we call this the {\it flow relaxation}. In Section~\ref{sec:flow_relax}, we construct instances that have a gap between the energy of the optimal tree and the flow relaxation of the order $\Theta(\sqrt{n}/\log n)$ over grid graphs and $\Theta(\Delta)$ in general graphs, where $n$ and $\Delta$ are the number of nodes and maximum degree in the graph (which can be linear in $n$) respectively. Further, we propose a randomized iterative edge-deletion algorithm, \textsc{RIDe}, that sparsifies the graph by deleting edges sampled according to a specific probability distribution dependent on the {\it effective} resistances of remaining edges. We show that this method can guarantee $\bigO(m-n)$ approximation with respect to the flow relaxation. This technique of sparsifying graphs may be of independent interest.
\end{itemize}

\begin{itemize}[leftmargin=*]    
    \item[(b)] {\it New Lower Bounds:}  We next exploit the combinatorial structure of the graph to obtain novel lower bounds in Section \ref{sec:n_and_logn}. We first show that any shortest-path tree, with respect to resistances, gives an $\bigO(n)$ approximation. \sg{Though a simple argument, this is already better than the $O(m^2-m)$ bound using \cite{Baes_2012} and the approximation using {\sc RIDe} for dense graphs}. Second, we show that we can improve this approximation factor by finding a laminar family of cuts. In particular, for \sg{certain} grid instances that have uniform edge resistances, a selection of laminar cuts gives an $\bigO(\sqrt{n})$-approximation. These results serve as preliminaries for our \minmin\ algorithm.
\end{itemize} 
 
\begin{itemize}[leftmargin=*]
    \item[(c)] {\it Constant-factor approximation using new \minmin\ algorithm:} \sg{Real-life distribution networks often resemble subgraphs of mesh-like networks. For such networks, like grid graphs with $n$ nodes and uniform demands, the above mentioned techniques, are able to provide only an $\Omega(\sqrt{n}/{\log n})$ approximation. Motivated by this, we construct a purely combinatorial} algorithm \minmin\, in Section \ref{sec:2apx}, that finds a specific shortest path tree over an $n \times n$ grid with uniform resistances and a root at one of the corners of the grid. We show that \minmin\ gives a $(2+\bigO(\frac{1}{\ln(n)}))(\frac{d_{\max}}{d_{\min}})^2$ approximation when the demands are in $[d_{\min}, d_{\max}]$. In particular, for uniform demands \minmin\ gives an asymptotic 2-approximation. 
  \end{itemize} 
\begin{itemize}[leftmargin=*]
    \item[(d)] {\it Layered Matching Heuristic and Computational Results:} Inspired by the algorithmic ideas in {\sc Min-Min}, we propose a layered matching heuristic called {LM}. \sg{This heuristic can be used to find approximate solutions in the most general setting, without assumptions on the structure of network, resistances or demands.} Using computational experiments over \sg{randomly} sparsified grid \sg{networks}, we find that {\sc LM} performs very well in practice. In addition, the algorithms proposed in this paper take orders of magnitude less time than the best known heuristic for this problem, the branch exchange heuristic, while obtaining comparable performance. For example, on $25 \times 25$ grid instances with sparsification probability $p = 0.2$, the mean time taken by LM is 1.35 seconds, whereas the mean time it takes the Branch Exchange heuristic to attain the same cost as LM is around 10 hours. \sg{We believe this improvement will be crucial in enabling} system operators to reconfigure distribution networks more frequently in practice. 
\end{itemize}

\section{Preliminaries for Electrical Flows}\label{sec:preliminaries}

\sg{Relaxing the support constraints in the network reconfiguration problem reduces the problem to computing an electrical flow, that we refer to as the {\it flow relaxation}.  Electrical flows have been shown to be efficiently computable in near-linear time \cite{Madry_2010,peng_2011,Madry_2016,cohen_2019}, used to speed up the computation of maximum flow \cite{Madry_2016}, 
and even bound the integrality gap of the asymmetric traveling salesman problem {(ATSP)} \cite{anari2015effective}.} We briefly review preliminaries on electrical flows in this section, and defer details to Appendix \ref{sec:background}. 

Let $G = (V,E)$ be a connected and undirected graph with $\abs{V}=n$, $\abs{E}=m$. Let $B \in \mathbb{R}^{n \times m}$ be the vertex-edge incidence matrix upon orienting each edge in $E$ arbitrarily. Let $R$ be an $m\times m$ diagonal resistance matrix where $R_{e,e} = r_{e}$. \sg{A key matrix that will play a fundamental role in the analysis of our algorithms is the weighted Laplacian} $L := BCB^T$, where $C = R^{-1}$. It is well known that if $G$ is connected, the only vector in the nullspace of the Laplacian $L$ is the all-ones vector $\mathbf{1}$. 

In what follows, we will invert the Laplacian \sg{matrix} using the Moore-Penrose pseudoinverse denoted by $L^\dagger$; in which case $L L^\dagger$ is a projection matrix that projects onto the span of the columns of $L$, which we denote by $\mathrm{im}(L)$. Let $b  \in \mathbb{R}^n$ be the feasible node-demand vector. The optimality conditions of the \sg{flow relaxation problem} imply the existence of a vector of potentials on the nodes (dual variables) $\phi  \in \mathbb{R}^n$ such that $\phi = L^\dagger b $ (this is well-defined since $\mathbf{1}^T b = 0$) and the optimal electrical flow $f = CB^T\phi$. Using these facts, one can show that the optimal energy $\E(f) = R^TfR =  \phi^T L \phi = b^T \phi = b^T L^\dagger b$. For any pair of vertices $u,v$ the effective resistance $\Reff(u,v)$ is the energy of sending \emph{one} unit of electrical flow from $u$ to $v$. In particular, for any vertex $u \in V$, if we let $\mathbbm{1}_u \in  \mathbb{R}^{n}$ be the characteristic vector of $u$, then $\Reff(u,v) = \chi_{uv}^T L^\dagger\chi_{uv}$, where $\chi_{uv} = \mathbbm{1}_v - \mathbbm{1}_u$ and can be thought of as the demand vector in this case. \sg{While the above notation suffices for our purposes, for more background on electrical flows, we refer the reader to \cite{williamson2019network} and \cite{lyons_2017}. We next discuss our novel randomized iterative edge-deletion algorithm, {\sc{RIDe}}, that rounds a fractional point (i.e., minimum energy flow, relaxing the support constraint in \eqref{eq:p1}) in the flow polytope, while maintaining a provable increase in the energy.}

\section{A Randomized Iterative Edge-deletion Algorithm} \label{sec:flow_relax}
The key idea of {\sc RIDe} is to delete edges iteratively following a specific probability distribution, while maintaining the graph connectivity, until the resultant graph is a spanning tree. This is done as follows: sample an edge $e$ at random to delete from the graph with probability $p_e$ proportional to $1 - c_e \Reff(e)$, where $c_e = {1}/{r_e}$ is the conductance of an edge. The normalization constant of this probability distribution \sg{is known to be} $\sum_{e \in E} (1 - c_e\Reff(e)) = m - (n - 1)$ \sg{(since $\big(c_e \Reff(e)\big)_e$ is in Edmonds' spanning tree polytope)}  \cite{Durfee_2016}
~(see Appendix \ref{sec:Uniform Spanning Trees}). Intuitively, the quantity $c_e\Reff(e)$ fully characterizes the graph's ability to \emph{reconfigure} the flow upon deleting edge $e$ (as we show later in Lemma \ref{recur}). The smaller the $c_e\Reff(e)$ (thus, the larger probability of deleting the edge), the better the graph's ability to re-route the flow of edge $e$ upon deleting the edge, without significantly increasing the energy. Moreover, this sampling procedure ensures that connectivity is maintained, since $p_e = 0$ for any bridge edge $e$, as in such a case we have $\Reff(e)=r_e$. To implement this algorithm efficiently, we show that the resultant graph Laplacian and effective resistances can be efficiently updated after every deletion, included in Algorithm \ref{alg:iterative deletion}. 


\begin{center}
\begin{algorithm}[t]
\caption{Randomized iterative deletion (\textsc{RIDe}) algorithm }
\label{alg:iterative deletion}
\begin{algorithmic}[1]
\INPUT A graph $G_0 = (V,E)$ and resistances $r : E \to \mathbb{R}_{++}$.\vspace{2 pt}
\For{$k = 0, \dots, m - n$}
        \State Compute $\Reffk{k} (e)$ for all $e \in E$, i.e., the effective resistance for each $e$ in $G_k$ 
        \State Sample edge $e \in E$ according to probabilities $p^{(k)}_e= (1 - c_e \Reffk{k}(e))/(m - k - (n - 1))$
        \State $G_{k+1} \leftarrow G_k \setminus\{{e}\}$ 
\EndFor
\vspace{5 pt}
 \RETURN Spanning tree $T = G_{m-n +1}$
\end{algorithmic}
\end{algorithm}
\vspace{-10 pt}
\end{center}


\sg{We will also show that} {\sc RIDe} gives an $\bigO(m-n)$ approximation factor in expectation with respect to the cost of the flow relaxation.  \sg{This is the first approximation guarantee for randomized edge deletion heuristics (such as \cite{shirmohammadi1989reconfiguration}), to the best of our knowledge.} 

\begin{theorem}\label{rand theorem}
    The randomized iterative edge-deletion algorithm, \textsc{RIDe}, gives an $\bigO(m-n)$ approximation in expectation with respect to the cost of the flow relaxation $\E(f_G)$ on the given graph $G$ {as $\mathbb{E}[\E(f_{ T})] \leq  \E(f_G)(m-n+2).$}
\end{theorem}

\sg{Note that the above theorem implies that for planar graphs, {\sc RIDe} gives an $O(n)$ approximation (since number of edges is linear in the number of nodes). In general, it seems one cannot obtain better than $\Omega(n)$ performance using \sg{rounding of} electrical flows, unless the flow relaxation is strengthened using new inequalities.}  
\sg{For instances with maximum degree $\Delta$, the gap from the flow relaxation can be $\Omega(\Delta)$. Consider a graph with two  nodes $r,t$ ($r$ is the root, and $t$ is the only node with positive demand, say 1 unit) with $n-2$ edge-disjoint paths between them; see Figure~\ref{fig:relax} (left). In this case, the electrical flow sends $1/\Delta$ units of flow along each of the $\Delta$ disjoint $r$-$t$ paths. The energy of the electric flow is then $\Theta(1/\Delta)$, however, the energy of the optimal spanning tree is $2$. Therefore, the gap from flow relaxation can be as large as $\Omega(n)$, since $\Delta$ can be as large as $\Theta(n)$.}

Moreover, the gap of the optimal tree compared to the flow relaxation can be large even for graphs with \sg{small $\Delta$}. Consider a $\sqrt{n} \times \sqrt{n}$ grid (with $n$ nodes) where the root $r$ is in the top left corner, all edges have unit resistances, the node in the bottom right corner, call it $t$, has a demand of one, and all other nodes (excluding the root) have zero demand; see Figure~\ref{fig:relax}~(middle) for an example {(here $\Delta=4$)}. Recall that the potential drop between $r$ and $t$ on sending \emph{one} unit of current from $r$ to $t$ is equal to the effective resistance $\Reff(r,t)$ between $r$ and $t$. Hence, using Ohm's Law this implies that the energy of the electrical flow is equal to $\Reff(r,t)$. In this instance, it is known that (see Proposition 10.11 in \cite{Levin_2006}) 
$$ \frac{1}{2} \log \sqrt{n} \leq  \Reff(r,t) \leq 2\log \sqrt{n}.$$
Furthermore, the cost of any optimal tree is $2 (\sqrt{n}-1)$, since any $r$-$t$ path has hop-length $2 (\sqrt{n}-1)$ and the path from $r$-$t$ can be grown into a spanning tree without incurring any additional cost. Combining these two facts, we get the gap for grid instances is $\Theta(\sqrt{n}/\log \sqrt{n}) = \Theta(\sqrt{n}/\log n)$.\sg{\footnote{We will improve upon this factor in Section \ref{sec:minmin}.}}

\begin{figure}
\centering
\begin{minipage}{.33\textwidth}
  \centering
  \includegraphics[scale = 0.40]{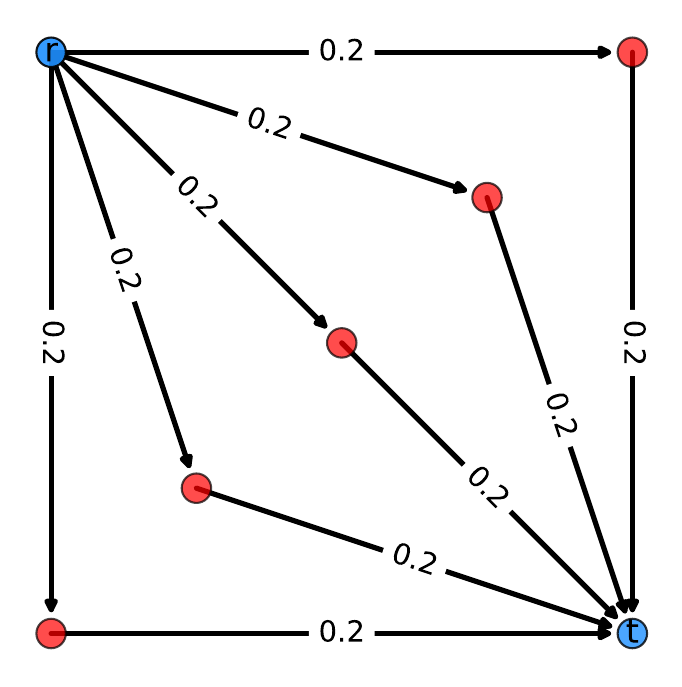}
\end{minipage}%
\begin{minipage}{.33\textwidth}
  \centering
  \includegraphics[scale = 0.40]{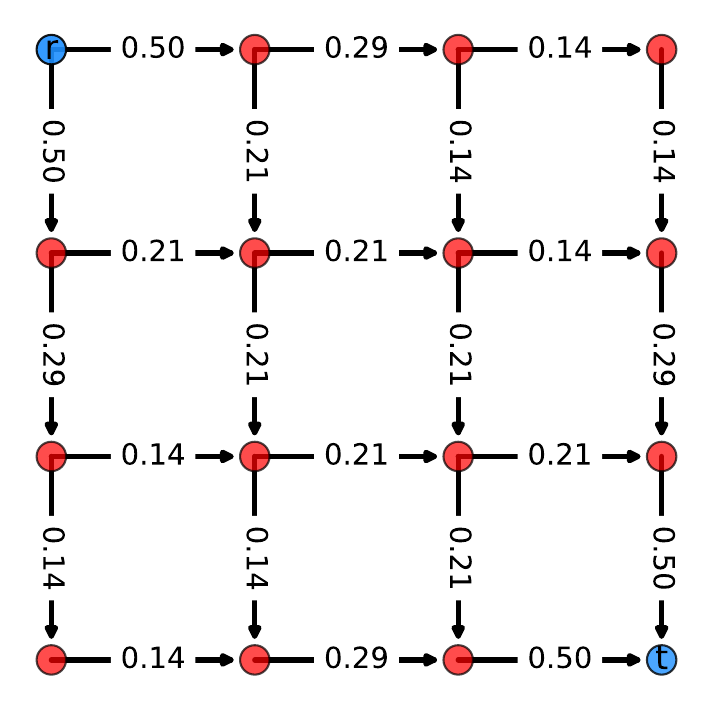}
\end{minipage}
\begin{minipage}{.33\textwidth}
  \centering
  \includegraphics[scale = 0.40]{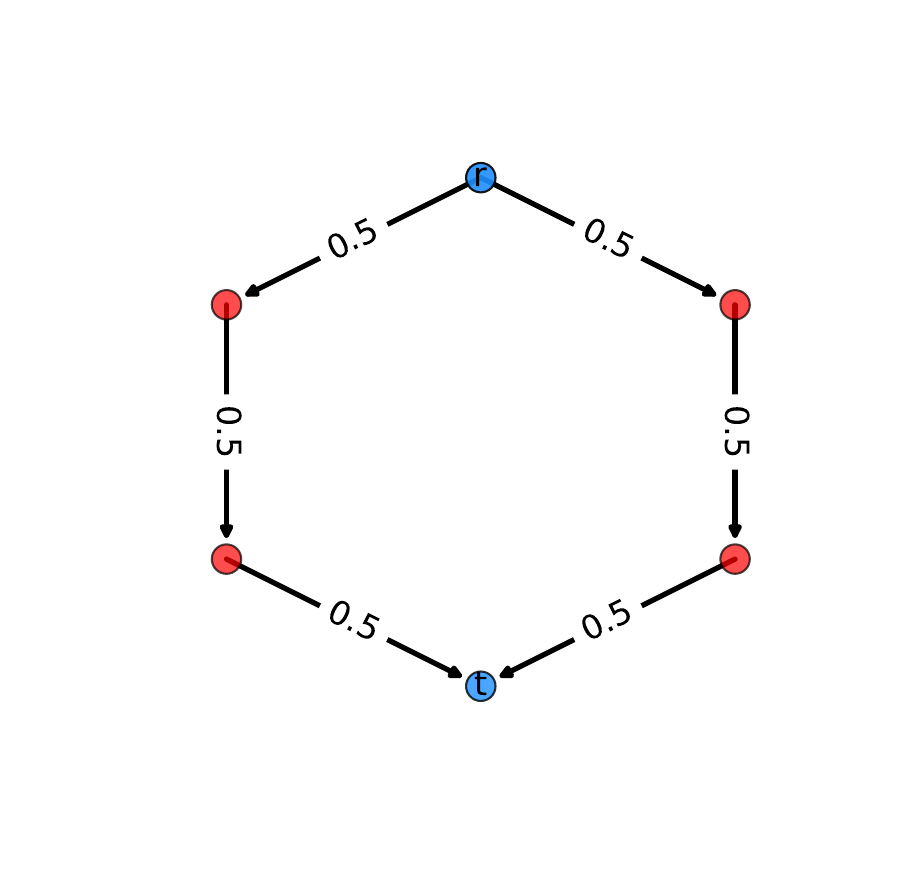}
\end{minipage}
\caption{The electrical flow on three graph instances, where all edges have unit resistance, the root is node $r$, node $t$ has unit demand and all other nodes (excluding $r$) have zero demand.}
 \label{fig:relax}
\end{figure}

Even existing work on analyzing the stretch\footnote{The stretch of a tree is a metric used to analyze how well a tree preserves distances between the endpoints of edges in the original graph.} of trees that has been used to bound the energy of a tree with respect to the flow relaxation \cite{kelner_2013}, does not give compelling approximation bounds, since there exist instances with {stretch $\Omega(m \log n)$} \cite{karp_1995}; we refer the reader to Appendix \ref{sec:relwork} for more details. {We next present our sparsification approach followed by the analysis of \textsc{RIDe}.} 

\subsection{Iterative Edge Deletions and Updates} \label{sec:iter del}
{Before we delve into the proof for the performance of {\sc RIDe}, we discuss one} of the main components in designing and analyzing the algorithm: determination of how the energy of the electrical flow changes after deleting an edge from the graph. \sg{As mentioned in preliminaries}, electrical flows are fully determined by the Laplacian and its pseudoinverse. Thus, to determine how electrical flows change upon edge deletions, we first obtain a closed form expression for how the Laplacian pseudoinverse changes from one iteration to the next, using the following extension of the Sherman-Morrison formula:
\begin{lemma}[Theorem A.70 in \cite{rao_1999}] \label{pseudo update} Suppose that $A \in  \mathbb {R} ^{n\times n}$ is symmetric matrix, $u,v\in \mathbb {R} ^n$ are vectors in $\mathrm{im}(A)$, and $1 - v^T A^\dagger u \neq 0$. Then, we have that  
$(A - uv^T)^\dagger  = A^\dagger + \frac{A^\dagger uv^TA^\dagger}{ 1 - v^T A^\dagger u}.$
\end{lemma}

We are now ready to prove the following result about updating the pesudoinverse of the Laplacian after an edge is deleted:

\begin{lemma}\label{lap update}
Let $L$ be the weighted Laplacian (weighted by conductances) of a connected graph $G = (V,E)$ and $L^\dagger$ be the Moore-Penrose pseudoinverse of $L$. Let $G^\prime = G \setminus \{e\}$ be the graph obtained by deleting an edge $e \in E$ that does not disconnect $G$. Then, the Moore-Penrose pseudoinverse of the weighted Laplacian of $G^\prime$ is:
\begin{equation} \label{new lap}
    (L')^\dagger = L^\dagger + \frac{L^\dagger \chi_e c_e \chi_e^T L^\dagger}{1 - c_e \chi_e^T L^\dagger\chi_e}.
\end{equation}
\end{lemma}

\begin{proof}
Let $e$ be the edge chosen for deletion. Observe that deleting $e$ from $G$ results in a rank-one update to $L$:
$$L^\prime = L - \chi_e c_e \chi_e^T = L - (\sqrt{c_e}\chi_e) (\sqrt{c_e}\chi_e)^T,$$
since $L = BCB^T$. {Using Lemma \ref{pseudo update}}, it suffices to show that $\sqrt{c_e}\chi_e \in \mathrm{im}(L)$ and $1 - c_e \chi_e^T L^\dagger\chi_e \neq 0$ (as $L \in \mathbb {R} ^{n\times n}$ is symmetric). First, since $C$ is a positive definite matrix we can write $L = BCB^T = (BC^{1/2}) (BC^{1/2})^T$. Let $U = BC^{1/2}$ and $u = \sqrt{c_e}\chi_e$. By construction $u$ is a column of $U$, {which implies that $u \in \mathrm{im}(U)$. Furthermore, since\footnote{{For any matrix $A \in \mathbb{R}^{m \times n}$, we have $\mathrm{im}(A) = \mathrm{im}(A A^T)$ (see e.g., Thm A.25 (iv) in \cite{rao_1999}}).} $\mathrm{im}(U) = \mathrm{im}(U U^T)$, it follows that $u$ is also in $\mathrm{im}(U U^T) = \mathrm{im}(L)$}. We further claim that $1 - c_e \chi_e^T L^\dagger\chi_e \neq 0$. To see this, suppose $1 - c_e \chi_e^T L^\dagger\chi_e = 0$. This implies $\chi_e^T L^\dagger\chi_e = 1/c_e$ or $\Reff(e) = r_e$. In other words, when we send one unit of electrical flow between the endpoints of $e$, all the flow goes through edge $e$. But this happens if and only if $e$ is a bridge (see Appendix \ref{intro to elec flows} for more details), which contradicts the assumption that $G^\prime$ was connected. Now, applying the pseudoinverse update formula given in {Lemma \ref{pseudo update}} yields the result.
\end{proof}

We can now use Lemma \ref{lap update} to obtain a closed form expression for how the energy increases from one iteration to the next. Suppose that we have performed $k$ iterations and deleted $k$ edges from the original graph $G$ to get a modified connected graph $G_k = (V, E_k)$ with $m-k$ edges. Let $B_k$, and $L_k$ respectively be the incidence and Laplacian matrices of $G_k$. Also, Let $f_k(e)$ and $\Reffk{k}(e)$ respectively be the electrical flow and effective resistance for edge $e$ in iteration $G_k$. Now, in the $(k+1)^\mathrm{th}$ iteration we wish to delete another edge $e$ from $G_{k}$ to obtain a connected graph $G_{k+1}$. We want to determine how much $\E(f_{k+1})$ changes compared to $\E(f_{k})$.

\begin{lemma} \label{recur}
  Assume $G_{k}$ is connected and $G_{k+1} = G_k \setminus \{e\}$ is obtained by deleting an edge $e$ from $G_k$ such that $G_{k+1}$ is also connected. Then $\E(f_{k+1})$ can be recursively obtained from $\E(f_{k})$ using: 
\begin{equation} \label{recursion1}
   \E(f_{k+1}) = \E(f_{k}) +\frac{r_e f_k(e)^2}{1 - c_e \Reffk{k}(e)}.
\end{equation}
\end{lemma}

\begin{proof}
Since $G_k$ and $G_{k+1}$ are connected, using Lemma \ref{lap update} we get: 
$$ (L_{k+1})^\dagger = L^\dagger_k + \frac{L^\dagger_k \chi_e c_e \chi_e^T L^\dagger_k}{1 - c_e \chi_e^T L^\dagger_k\chi_e}.$$
Therefore, the new potentials are given by
$$\phi_{k+1} = L_{k+1}^\dagger b \\
= L^\dagger_k b + \frac{L^\dagger_k \chi_e c_e \chi_e^T L^\dagger_k b}{1 - c_e \chi_e^T L^\dagger_k\chi_e}\\
= \phi_k + \frac{L^\dagger_k \chi_e c_e \chi_e^T \phi_k}{1 - c_e \chi_e^T L^\dagger_k\chi_e}= \phi_k + \frac{L^\dagger_k \chi_e f_k(e)}{1 - c_e \chi_e^T L^\dagger_k\chi_e},$$
where we used the fact that $f_k = C_kB_k^T\phi_k$ (i.e. $f_k(e) = c_e \chi_e^T \phi_k$) in the last equality. By connectivity and feasibility assumptions we know that $b \in \mathrm{im}(L_k)$ and $\chi_e \in \mathrm{im}(L_k)$. This implies that $L_k \phi_k = b$ and similarly, $L_k L^\dagger_k \chi_e = \chi_e$. Then, using $\E(f_{k+1}) = b^T \phi_{k+1}$, we have
\begin{align*}
     \E(f_{k+1}) = b^T \phi_{k+1} &= \E(f_k) +  \frac{ b^T L^\dagger_k \chi_e f_k(e)}{1 - c_e \chi_e^T L^\dagger_k\chi_e}\\
     &\overset{\small (a)}{=} \E(f_{k}) +\frac{(L_k \phi_k)^T L^\dagger_k \chi_e  f_k(e) }{1 - c_e \Reffk{k}(e)}
     \overset{\small (b)}{=}  \E(f_{k}) +\frac{ \phi_k^T \chi_e f_k(e)}{1 - c_e \Reffk{k}(e)},
\end{align*}
where (a) follows from the fact that $L_k \phi_k = b$ and (b) from $L_k L^\dagger_k \chi_e = \chi_e$. The result then follows  from the optimality conditions $\chi_e^T \phi_k = f_k(e)/ c_e =r_e f_k(e)$.
\end{proof}

We would like to remark that the expression above holds for any edge $e \in E$ we delete from the graph as long as $e$ does not disconnect the graph as previously mentioned. We next discuss the performance of {\sc RIDe}, where an edge $e$ is deleted by randomly sampling it proportional to $1-c_e\Reff(e)$.

\subsection{Performance of {\sc{RIDe}} Algorithm} \label{Ride alg} 
 We are now ready to present the proof for Theorem \ref{rand theorem}. 
 
 \begin{proof} Let $G_k=(V, E_k)$ be the resultant graph in iteration $k$ after $k$ edges have been deleted, and let $\E(f_k)$ be the energy of the electrical flow in $G_k$. In particular, $\E(f_0)$ denotes the cost of the flow relaxation. We will first show 
    \begin{equation} \label{induc}
        \mathbb{E}[\E(f_{k})] \leq  \E(f_0) \left( \frac{m - n + 2}{m-k-n +2}\right), 
    \end{equation}
    and that $G_k$ remains connected throughout the iterations of the algorithm by induction on $k$ ($0 \leq k \leq m-n+1$). The base case when $k = 0$ holds trivially. For the inductive step assume the result holds true for all iterations $k < m-n + 1$. Now consider iteration $k + 1$ of the \textsc{RIDe} algorithm applied to the graph $G_k$. Recall that we sample an edge $e \in E_k$ at random to delete from the graph with probability $p^{(k)}_e = \frac{1 - c_e \Reffk{k}(e)}{(m-k) - (n - 1)}$. Since by the induction hypothesis $G_k$ is connected, and this sampling procedure does not disconnect the graph, $G_{k+1}$ is also connected. Moreover, using Lemma \ref{recur} to compute the expected increase in energy upon deleting edge $e$ (which does not disconnect the graph), we get
\begin{equation*} \label{formula}
    \E(f_{k+1}) = \E(f_{k}) +\frac{r_e f_k(e)^2}{1 - c_e \Reffk{k}(e)}.
\end{equation*}
Let $Z^{(k)}$ be a random variable denoting the increase in the energy from iteration $k$ to $k+1$. Upon deleting an edge $e \in E_k$ that does not disconnect $G_k$, we have $Z_e^{(k)} = r_e f_k(e)^2/(1 - c_e \Reffk{k}(e))$. Also, let $E_k^\prime = \{e \in E_k \mid p^{(k)}_e >0 \}$ be the set of edges that are not bridges. Then, since $E_k^\prime \subseteq E_k$, we have
\begin{align*}
    \mathbb{E}[Z^{(k)} \mid E_k] = \sum_{e \in E_k^\prime} Z_e^{(k)} p^{(k)}_e 
    &= \sum_{e \in E_k^\prime} \frac{r_e f_k(e)^2}{1 - c_e \Reffk{k}(e)}  \frac{1 - c_e \Reffk{k}(e)}{(m-k) - (n - 1)}\\
    &\leq \frac{\E(f_{k})}{m - k - n + 1}, 
\end{align*}
which in turn gives, using iterated expectations:
$$\mathbb{E}[Z^{(k)}] = \mathbb{E}[\mathbb{E}[Z^{(k)} \mid E_k]] \leq  \frac{\mathbb{E}[\E(f_{k})]}{m - k - n + 1}.$$
Therefore,
\begin{align*}
\mathbb{E}[\E(f_{k+1})] &= \mathbb{E}[\E(f_{k}) + Z^{(k)}] \leq  \mathbb{E}[\E(f_{k})] \left(1 + \frac{1}{m-k-n +1}\right)\\
&\leq {\E(f_0)  \left( \frac{m -n + 2}{m-k-n +2}\right) \left( \frac{m -k-n + 2}{m-k-n +1}\right)} \\
&= {\E(f_0) \left( \frac{m -n + 2}{m-(k+1)-n +2}\right)}, 
\end{align*}
where we used the induction hypothesis in the second inequality. This concludes the induction and proves the correctness of the algorithm. Lastly, to obtain the final expected cost of the algorithm, we use \eqref{induc}  with $k = m- n + 1$ to obtain $\mathbb{E}[\E(f_{ m- n +1})] \leq \E(f_0)(m-n+2)$, as claimed. {Since the cost of the optimal spanning tree is lower bounded by the cost of the flow relaxation, $\E(f_0)$, the result follows.}
\end{proof}



\sg{The approximation factor for {\sc RIDe} proved above is in fact tight, as can be seen from the following simple example}. Consider a cycle of $n$ nodes, where $n$ is even, fix a root $r$ arbitrarily and assume all the resistances are one. Now assign the node of hop-length $m/2$ from the root, call it $t$, a unit demand and all other nodes (excluding the root) zero demand. Then, it is easy to see that the flow relaxation will send a flow of 0.5 from $r$ to $t$ along the two disjoint $r$-$t$ paths of hop-length $m/2$ in the cycle; see Figure~\ref{fig:relax}~(right) for an example. Thus, the cost of the flow relaxation is ${m/4}$. Moreover, the optimal tree will just send a flow of 1 from $r$ to $t$ along on a path of hop-length $m/2$, which implies the cost of the optimal tree is $m/2$. On the other hand, since each edge in the cycle has the same effective resistance, using the \textsc{RIDe} algorithm, each edge is equally likely to be deleted. Using Theorem \ref{rand theorem} and noting that $m = n$ in this case, we obtain a spanning tree whose cost is a 2-approximation from the flow relaxation in expectation. This matches the gap between the optimal spanning tree and the flow relaxation.

In light of this discussion, it is imperative to note that our approximation factor is with respect to the \emph{flow relaxation}, which is a loose lower bound as discussed previously. For planar graphs in particular, the gap of optimal from the flow relaxation can be  $\Omega(n)$ (e.g., when the maximum degree is linear), thus in some sense {\sc RIDe} is optimal {up to} a constant factor for the planar case. However, this does not preclude the possibility of obtaining a better lower bound and approximation using electrical flows and this remains an open question. To strengthen the lower bound and consequently obtain better approximation factors, we proceed by exploiting the combinatorial structure in certain graphs as well as demand scenarios.

\section{New Lower Bounds}\label{sec:n_and_logn}

As discussed in Section \ref{sec:flow_relax}, relaxing the spanning tree constraint can lead to a {weak} lower bound, mainly because we allow the demand of a single node $v$ to be delivered via multiple paths from root $r$ to node $v$. \sg{A natural question at this point is if we can strengthen the flow relaxation by exploiting the tree constraints.} We first answer this question by accounting for the minimum loss each node $v$ creates to get connected to the root, in the absence of all other nodes, and show how to use this lower bound to achieve an $n$-approximation algorithm. Next, we consider cuts in a graph to lower bound the energy of an optimal tree by considering the demand it separates. We show that by constructing a laminar family of cuts, one can derive a new lower bound by accounting for the loss of a spanning tree which is balanced over all these cuts. We then use this lower bound to get a $\sqrt{n}$-approximation for grid graphs.
\subsection{Shortest-path Trees}

The major source of hardness in \eqref{eq:p1} is the quadratic loss function, which introduces a cross-term for any two nodes that share an edge on their path to the root. If the loss was a linear function of the flow, 
and in the absence of
cross-terms, the problem would decompose into $n$ disjoint problems, which {could} be solved via shortest path trees. 
Note that, however, we can relate the quadratic loss to the linear case as follows:
\begin{equation}\label{eq:CS}
    \sum_{i\in \text{succ}(e)} d_i^2\leq \left( \sum_{i\in \text{succ}(e)} d_i\right)^2\leq n\times \sum_{i\in \text{succ}(e)} d_i^2,
\end{equation}
where the first inequality is by the non-negativity assumption, and the second one is due to the Cauchy-Schwarz inequality. Now looking at $d_i^2$ as the new demand of each node $i$, our quadratic loss lies between the two linear objectives, for which we have to solve the following optimization problem: $\min_{T\in \mathcal{T}} \sum_{e\in T} \left[r_e\times \sum_{i\in \text{succ}(e)} d_i^2\right]$. It is easy to show that the shortest-path tree (SPT) rooted at node $r$ (with respect to edge resistances, $r_e$) solves this optimization problem\footnote{See the problem \textit{SymT} in \cite{gupta2001provisioning} for example.}, which immediately implies the following result.
\begin{theorem}\label{thm:spt}
The shortest-path tree (with respect to resistances) rooted at $r$  is an $n$-approximation solution for problem \eqref{eq:p1}. Moreover, there exist graph instances for which the cost of a shortest-path tree (or BFS tree) is at least $\Omega(n)$ times the cost of an optimal tree.
\end{theorem}
\begin{proof}
The approximation factor follows from summing up equation \eqref{eq:CS} across all edges, and the fact that the shortest-path tree simultaneously minimizes both the costs in the left and right-hand side of the resulting inequality. For the lower bound, consider the graph shown in Figure~\ref{fig:BFS}. There are $n$ triplets of nodes in parallel, node $r$ as the root, and a final node labeled $3n + 1$. All nodes (except the root) have demand of $d_i=1$, and all resistances are equal. On the left, we have the shortest-path tree (BFS tree) whose cost 
can be calculated as:
$APX = n\times (1^2 +2^2)+(n+1)^2 +n\times 1^2= n^2 +8n + 1$.
On the right, we have the optimal tree which does not change the first $3$ triplets, but re-configures the rest as shown. The cost of this tree is 
$OPT =(n-3)\times (1^2 +2^2 +3^2)+3\times (1^2 +2^2)+4^2 +3\times 1^2 =14n-8$.
Comparing the two costs proves a lower bound of $\Omega(n)$ on the performance of SPTs.
\begin{figure}[t]
\centering
\begin{minipage}{.62\textwidth}
  \centering
  \includegraphics[width=0.96\linewidth]{fig2.pdf}
  \captionof{figure}{Lower-bound example on the performance of the shortest-path tree: (Left) shortest-path tree, versus (Right) optimal spanning tree.}
  \label{fig:BFS}
\end{minipage}%
\hspace{7pt}
\begin{minipage}{.32\textwidth}
  \centering
  \includegraphics[width=.9\linewidth]{fig3.pdf}
  \captionof{figure}{$n\times n$ grid with root at the top-left.}
  \label{fig:notation}
\end{minipage}
\end{figure}
\end{proof}


\subsection{Cut-based Lower Bounds}
{For the rest of this section, we assume that all the edges of the graph have the same resistances, and without loss of generality we set $r_e=1$ for all $e\in E$, and consider approximations using sets of cuts.} 


\begin{theorem}\label{thm:general_cuts}
Consider a graph $G=(V,E)$ with root $r\in V$, and $r_e=1, \forall e\in E$. Assume that $G$ has a family of cuts $S_1, S_2,..., S_\ell\subset V$ (with $r \in S_i, \forall i$) s.t.:
\begin{itemize}
    \item[(a)] each edge $e\in E$ appears in at most $M$ cuts, i.e., $\abs{\{i:e\in \delta(S_i)\}}\leq M$ for all $e\in E$,
    \item[(b)] the union of the cuts $\cup_{i} \delta(S_i)$ supports a spanning arborescence\footnote{An $r$-arborescence is a directed spanning tree such that for any vertex $v$, there is exactly one directed path from $r$ to $v$ (see \cite{schrijver2003combinatorial} for more details).} $A$ rooted at $r$ such that (directed) edge $(u,v)\in A$ only if $u\in S_i, v\notin S_i$ for all $i\in[\ell]$ such that $e\in \delta(S_i)$.
\end{itemize}
 Then, the approximation factor of the arborescence $A$ is at most $M\times \max_{i} |\delta(S_i)|$.
\end{theorem}

\begin{proof}
Let $K=\max_i \abs{\delta(S_i)}$ be the size of the biggest cut. We can use assumption (b) of the theorem statement, to map any edge $(u,v)$ of the arborescence to a cut $S_i$ (if there is more than one cut, we can pick one arbitrarily). In this way, we split the cost of arborescence $A$ among different cuts, while ensuring that the edges are carrying flows directed out of the cuts. Considering any of these cuts, we show that the costs of the trees $A$ and $OPT$ restricted to that cut are within a factor $K$, i.e., $\sum_{e \in \delta(S_i)} \mathcal{E}_{A}(e) \leq K \sum_{e \in \delta(S_i)} \mathcal{E}_{OPT}(e)$ for all $i$.

Let $S$ be one of the cuts with $r\in S$, and let $k=\abs{\delta (S)}$ be the size of the cut. Let $a_1,...,a_{k}$ be the flow on the edges crossing the cut {(calculated based on tree $A$)}, where $a_i>0$ if the flow is going out of the cut and $a_i<0$ if it is flowing inwards. With this choice of directions, the arborescence $A$ has only non-negative $a_i$ values. We also know that the total flow going across this cut is equal to the total demand below the cut, i.e., $\sum_{i=1}^{k} a_i=\sum_{v\notin S}d_v$.

Let $b_1,...,b_{k}$ be the flow of the edges crossing this cut in the optimal tree, where some of these variables may be zero if the optimal tree is not using that edge, or even negative if they are being used in the opposite direction. However, we still know that $\sum_{i=1}^{k}b_i=\sum_{v\notin S}d_v$, and also the cost of this cut in the optimal tree is $\sum_{i=1}^{k} b_i^2$. This cost is minimized when all the $b_i$'s are equal ($b_i=\frac{1}{k}\sum_{v\notin S}d_v$ gives a lower bound on the cost even if it is not attainable by any spanning tree.) Therefore we get:
\begin{equation}\label{eq:cut_size_argument}\small
    \sum_{e \in \delta(S)} \mathcal{E}_{OPT}(e)=\sum_{i=1}^{k} b_i^2\geq k\left(\frac{1}{k}\sum_{v\notin S}d_v\right)^2=\frac{1}{k}\left(\sum_{i=1}^k a_i\right)^2\geq \frac{1}{K}\sum_{i=1}^{k} a_i^2=\frac{1}{K}\sum_{e \in \delta(S)} \mathcal{E}_{A}(e),
\end{equation}
where in the last inequality we dropped the (non-negative) cross-terms $a_ia_j$, and used the fact that the cut size $k\leq K$. If the cuts were disjoint, the $K$-approximation would directly extend to the entire objective function as well, as the cuts would divide the objective into separate additive objectives. 
However, we may double count what the optimal tree is paying since the cuts are not disjoint, but we know that each edge will be counted at most $M$ times. This gives the approximation ratio of $M\times K$ in total.
\end{proof}


We use Theorem \ref{thm:general_cuts} for an $n\times m$ grid {(where we let the number of nodes be $N = mn$)}, and consider diagonal cuts as shown in Figure~\ref{fig:notation}. Note that these cuts partition the edges of the grid, and their size is less than $2\sqrt{N}$. Any spanning tree that has edges only going to the right or the bottom, satisfies the second requirement of Theorem~\ref{thm:general_cuts}, and thus has cost at most $2\sqrt{N}$ times the optimal spanning tree.

\begin{corollary}\label{thm:rootn}
There exists an $\bigO(\sqrt{N})$-approximation algorithm for minimizing the loss on an $n\times m$ grid with $N$ nodes, when all the edges have the same resistance and the root is located at the corner of the grid.
\end{corollary}

Constructing the above described set of cuts for general graphs remains an open question; planar graphs would be a natural candidate. By the planar separator theorem {\cite{lipton1979separator}}, we know that any $N$-node planar graph has a vertex separator of size $\bigO(\sqrt{N})$ that splits the graph into two (almost) equal parts. However, it is not clear how to find the desired family of cuts by using the planar separator oracle. This is due to the second requirement of the cuts in Theorem~\ref{thm:general_cuts} which fixes a natural direction on any edge once it appears in a cut. This direction should be respected in future cuts that include this edge; however, the separator oracle is oblivious to edge directions. See Appendix~\ref{app:planar} for a discussion of the challenges of generalizing these results to planar graphs.

\section{A 2-Approximation for Uniform Grid}\label{sec:2apx}
We now propose a constant-factor approximation algorithm for an $n\times m$ uniform grid ($n\leq m$) with the root at a corner of the grid (see Figure \ref{fig:notation}), in which all demands are equal ($d_i=1$ for all $i \in V\backslash \{r\}$) and all resistances are equal ($r_e=1$ for all $e\in E$). The key idea of considering this case is to help us better understand the structure of optimal solutions through combinatorial techniques, \sg{that can be generalized for real-life distribution networks (Figure \ref{fig:realnetwork}). Even though the demands and resistances are uniform in the $n\times n$ grid,} it is non-trivial to connect the loads together in a way that avoids big flow values close to the root. For example, Figure~\ref{fig:notation} demonstrates an example tree on an $n \times n$ grid, that satisfies the properties of both Theorem \ref{thm:spt} and Theorem \ref{thm:general_cuts} solutions, and yet fails to provide a constant-factor approximation. The horizontal edges in this tree have flows of $n,2n,...,(n-1)n$, and therefore the cost of the tree is in the order of {$\sum_{i=1}^{n-1} (i\times n)^2=\Theta(n^5)$ (vertical edges have a total cost of $\bigO(n^4)$). On the other hand, the optimal cost in this grid is $\Theta(n^4\log n)$, as we will show an $\Omega(n^4\log n)$ lower bound in Lemma~\ref{lemma:n4logn} and prove a constant-factor approximation (with respect to the same lower bound) in the rest of this section.}
For the sake of brevity, we will  present our novel {\sc Min-Min} algorithm for a square $n\times n$ grid, while the results hold for rectangular grids as well and more general demands.

\begin{figure}[t]
    \centering
  \includegraphics[width = 0.4\textwidth]{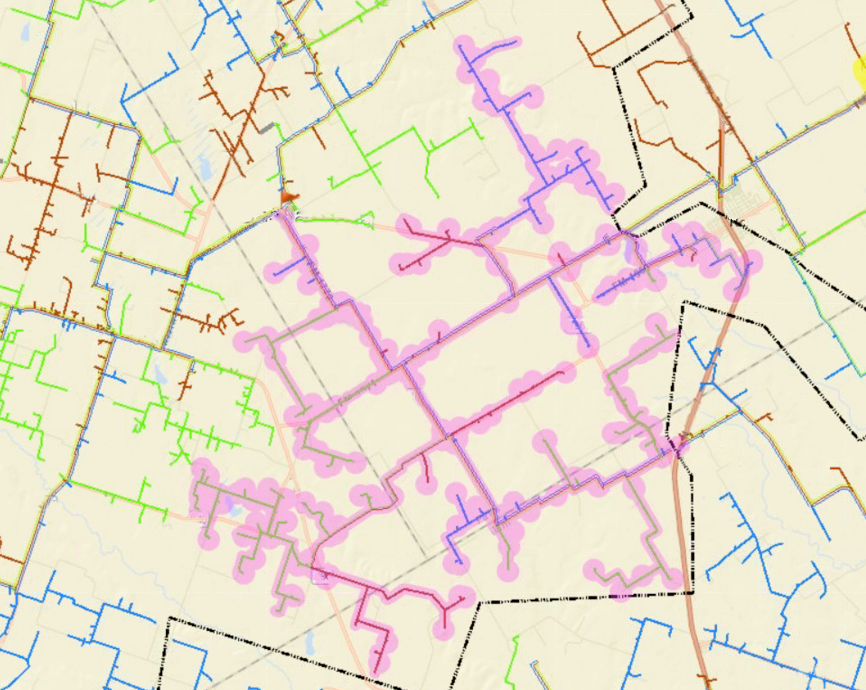}
    \includegraphics[scale = 0.9]{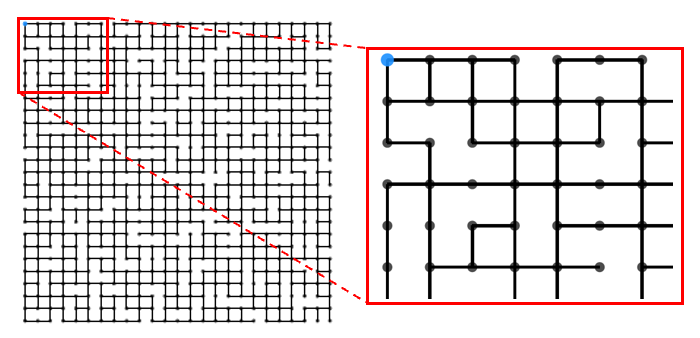}
\caption{Left: Snapshot of an anonymized real distribution network from a utility company in the US. These networks typically look like sparse subgraphs of grids, thus motivating our exploration of approximations in Section~\ref{sec:2apx} and computational experiments in Section \ref{sec:lm}. Right: An abstraction of a real-life distribution network through an instance of a sparsified $25\times25$ grid that was used in the computations, where the sparsification probability is $p = 0.2$ and the root is in the top left corner.}
\label{fig:realnetwork}
\end{figure}

\vspace{10pt}
\noindent {\bf Notation:} We let $n{\times n}$ be the size of the grid \sg{and the} total number of nodes be $N=n^2$. We consider the diagonal cuts as shown in Figure~\ref{fig:notation} and we name them $(u,i)$, $i=1,...,n-1$ for the upper triangle, and $(\ell,i)$ for the lower triangle. Note that the diagonal cuts cover all the edges and each edge appears in only one cut. Therefore, we can divide the cost of any spanning tree (either optimal or approximate tree) into the costs from each cut. Let $OPT^{(u,i)}$ and $Alg^{(u,i)}$ denote the cost of edges that cross cut $(u,i)$ in the optimal and approximate solution, respectively. Similarly we define $OPT^{(\ell,i)}$ and $Alg^{(\ell,i)}$ for the lower triangle. Finally, let $OPT^u=\sum_{i=1}^{n-1} OPT^{(u,i)}$ and $OPT^\ell=\sum_{i=1}^{n-1} OPT^{(\ell,i)}$. Then, we have $OPT=OPT^u+OPT^\ell$. Similarly, we define $Alg^u$ and $Alg^\ell$.

In Section~\ref{sec:lb}, we will use these diagonal cuts to find a lower bound for the cost of any spanning tree, explain our \minmin\ algorithm in Section~\ref{sec:minmin} followed by its performance guarantee in Section~\ref{sec:apx_analysis}.

\subsection{Lower bounds}\label{sec:lb}
Let $S_{u,i}$ and $S_{\ell,i}$ be the number of nodes below cuts $(u,i)$ and $(\ell,i)$, respectively. Then, 
$S_{u,i}=n^2-\sum_{j=1}^i j=n^2-\frac{i(i+1)}{2},$ and $ S_{\ell,i}=\sum_{j=1}^i j=\frac{i(i+1)}{2}$.

\vspace{10pt}
\noindent {\bf Upper triangle cuts:} By a quick look at the structure of the grid, we can observe that there are $2i$ edges that cross the cut $(u,i)$. These edges are connected to $i$ nodes on the root side of the cut, and $i+1$ nodes on the other side. Call these nodes $u_1,...,u_i$ on the root side and $v_1,...,v_{i+1}$ below the cut. We call an edge $(u_j,v_k)$ of the tree \emph{outgoing}, if $u_j$ is the parent of $v_k$ in the tree. We claim that the tree can have at most $i+1$ outgoing edges over this cut, although it can have all the $2i$ edges. This is because if we have more than $i+1$ outgoing edges, then by the pigeonhole principle, a node $v_k$ will have two parents and this creates a loop in the tree. 

In addition, all the $S_{u,i}$ nodes below the cut are connected to the root through (at least) one of these outgoing edges over this cut. This is true because we can traverse the path from the root to that node, and at some point we must cross the cut through an outgoing edge. It is possible to have multiple outgoing edges on that path if we cross the same cut multiple times, but all we need is that each node below the cut is counted as a successor for at least one of the outgoing edges. So the aggregate number of successors for the outgoing edges of cut $(u,i)$ is at least $S_{u,i}$, while there are at most $i+1$ such edges. Recall that when the summation of a number of variables is fixed, their sum of squares is minimized when all of them are equal. Hence, in the most balanced way any tree (including the optimal tree) has to pay the following cost over this cut:
\begin{equation}\label{eq:lb_uppercuts}
OPT^{(u,i)}\geq (i+1)\times \left(\frac{S_{u,i}}{i+1}\right)^2=\frac{S_{u,i}^2}{i+1}.
\end{equation}

{By summing \eqref{eq:lb_uppercuts} over different cuts, we get the following lower bound on the energy of any spanning tree over the entire upper triangle.} 
\begin{lemma}\label{lemma:n4logn}
The cost of the optimal tree over the upper triangle part of the grid is lower bounded by:
{$OPT^u \in  \Omega(n^4\log{n})$.}
\end{lemma}
\begin{proof}
To obtain a lower bound for $OPT^u$, we just plug the value of $S_{u,i}$ in \eqref{eq:lb_uppercuts} and sum over $i$:
\begin{align*}
OPT^u&\geq \sum_{i=1}^{n-1}\frac{S_{u,i}^2}{i+1}=\sum_{i=1}^{n-1}\frac{\left(n^2-i(i+1)/2\right)^2}{i+1}\nonumber\\
&= \sum_{i=1}^{n-1}\frac{n^4}{i+1}+\frac{1}{4}\sum_{i=1}^{n-1}i^2(i+1)-\sum_{i=1}^{n-1}n^2i\nonumber\\
&\geq n^4\big(\ln (n+1)-1\big)+\frac{(n-1)^2n^2}{16}+\frac{(n-1)n(2n-1)}{24}-\frac{n^3(n-1)}{2}\nonumber\\
&=n^4\ln (n+1)-\frac{23}{16}n^4+\frac{11}{24}n^3-\frac{1}{16}n^2+\frac{1}{24}n. \hspace{90 pt}
\end{align*}
\end{proof}
\vspace{10pt}
\noindent {\bf Lower triangle cuts:} The argument here is exactly the same as in the upper triangle except that there are $i+1$ nodes on the root side and $i$ nodes on the other side. Therefore, in the most balanced case when all the outgoing edges carry the same flow, we have $i$ outgoing edges with flows $S_{\ell,i}/i=(i+1)/2$, paying the total cost of:
\begin{equation}\label{eq:lb_lowercuts}
OPT^{(\ell,i)}\geq i\times \left(\frac{S_{\ell,i}}{i}\right)^2=\frac{i(i+1)^2}{4}.
\end{equation}

\subsection{{\sc \minmin } algorithm}\label{sec:minmin}
The \minmin\ algorithm builds a spanning tree which contains $n$ disjoint paths with different lengths over the lower triangle; see the blue paths in Figure~\ref{fig:minmin} (right). Then, in each cut of the upper triangle, exactly one pair of subtrees merge together. As the name suggests, we merge the two subtrees with the minimum number of successors in each step, and call it the \emph{merging step}. However, this requires those two subtrees to be next to each other. 
\begin{figure}[t]
\centering
\includegraphics[width=\linewidth]{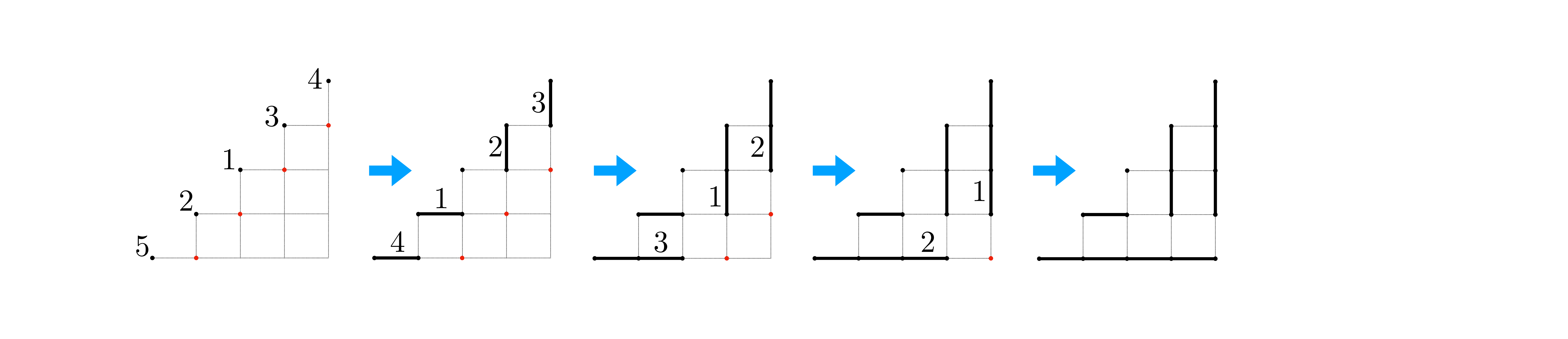}
\caption{Example of Lemma \ref{lemma:permutation}. In each step we decrease the number of successors desired by one and put them on the next diagonal (red nodes), except 1 which is already satisfied.}
\label{fig:permutation}
\end{figure}
Therefore, we need to order our disjoint paths in the lower triangle in a way that allows merging minimum load subtrees in the upper triangle; we call this the \emph{uncrossing step}. In the following lemma, we show that the number of successors of edges on the main diagonal can be any permutation of the numbers $1,2,...,n$.

\begin{lemma}[Disjoint paths]\label{lemma:permutation}
We can obtain a disjoint path decomposition of the lower triangle of the grid ({i.e., a set of $n$ paths from the $n$ points on the diagonal that are shortest between the end-points and cover all the vertices in the lower triangle}) for any ordering of numbers $1,2,...,n$, specifying the number of successors of edges on the left diagonal of the grid. 
\end{lemma}
\begin{proof}
We give a recursive construction which also proves the existence of such paths. Let $(a_1,a_2,...,a_n)$ be a permutation of $(1,2,...,n)$. Put these numbers on the main diagonal. Except $a_i=1$ which is already satisfied, connect the rest of the nodes to the nodes of {the} next diagonal, which has $n-1$ nodes, in the same order. Now we have to construct a permutation of $(1,2,...,n-1)$ on this new diagonal, because the previous numbers should be decreased by one. We can repeat the process. An example is performed in Figure~\ref{fig:permutation}.
\end{proof}

\begin{figure}[t]
\centering
\begin{minipage}{.25\textwidth}
  \centering
  \includegraphics[width=.9\linewidth]{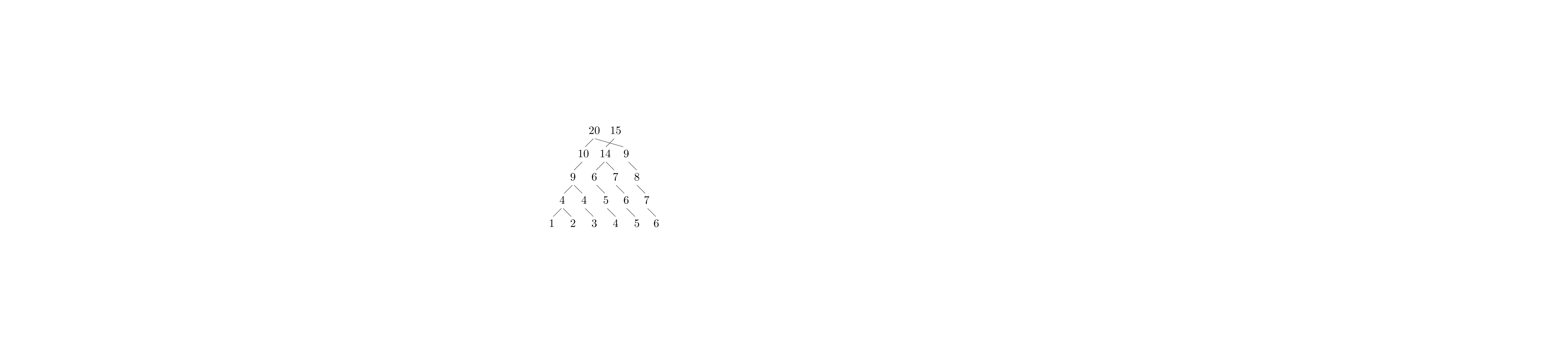}
\end{minipage}%
\begin{minipage}{.25\textwidth}
  \centering
  \includegraphics[width=.9\linewidth]{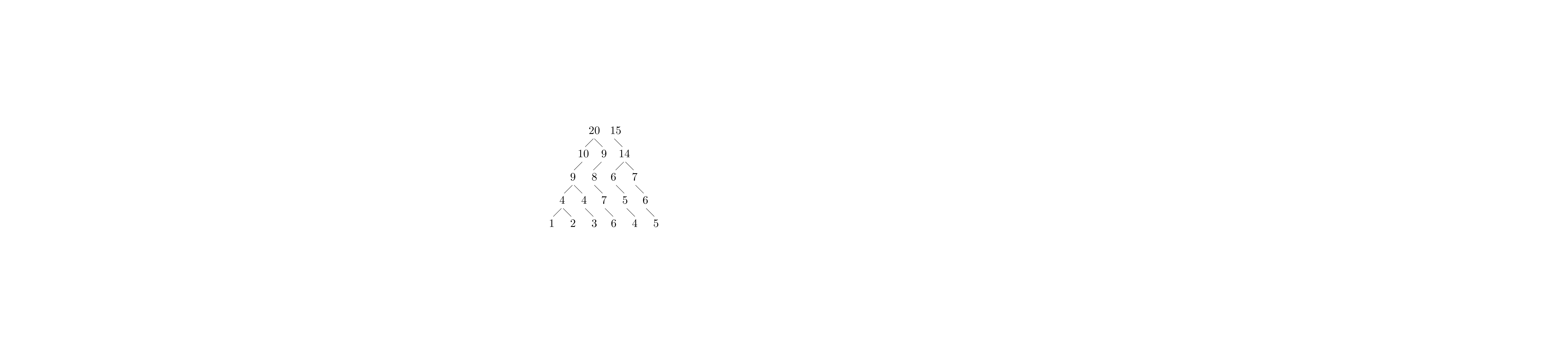}
\end{minipage}%
\begin{minipage}{.3\textwidth}
  \centering
  \includegraphics[width=.9\linewidth]{fig6c.pdf}
\end{minipage}
\caption{Example of the \minmin\ algorithm. {\bf Left}: Merging the two smallest numbers in each layer, starting from the $1,2,...,n$ sequence. {\bf Middle}: Same tree re-ordered from top to bottom to avoid crossings. {\bf Right}: The corresponding grid where the lower triangle is constructed by Lemma~\ref{lemma:permutation}, and the numbers in the upper triangle are merged in each diagonal layer according to the middle tree.}
\label{fig:minmin}
\end{figure}

Note that Lemma~\ref{lemma:permutation} ensures the adjacency of minimum subtrees in all upper triangle cuts. To get the right permutation, we can start from any permutation (say $1,2,...,n$) and do the \minmin\ merging as shown in Figure~\ref{fig:minmin}~(left). Then we can do the uncrossing from top to bottom as shown in Figure~\ref{fig:minmin}~(middle) and this gives the desired permutation on the main diagonal. Finally, we can construct the lower part of the spanning tree corresponding to that permutation using Lemma~\ref{lemma:permutation}, and construct the upper part of the spanning tree by merging minimum size subtrees in each layer. An example of the algorithm is shown in Figure~\ref{fig:minmin}~(right). The formal description of the algorithm for general rectangular grids is also tabulated under Algorithm~\ref{alg:minmin}, in which for the case of a rectangular grid, we use the middle part of the grid to connect the lower triangle to the upper triangle via parallel disjoint paths.
\subsection{Approximation factor for the {\sc \minmin } algorithm}\label{sec:apx_analysis}

\begin{algorithm}[t]\label{minmin} 
\caption{\minmin\ algorithm}
\label{alg:minmin}
\begin{algorithmic}[1]
\INPUT An $n\times m$ grid ($n\leq m$).\vspace{5 pt}
\State Start with sequence $(m-n)+1,(m-n)+2,...,(m-n)+n$.
\While{there are more than two numbers}
\State Merge the two smallest numbers.
\State Add one to the entire sequence.
\EndWhile
\State {\bf Uncrossing}: Backtracking the previous step, find the right permutation of the starting sequence to put on the diagonal cut $(u,n-1)$, to ensure the adjacency of smallest two numbers in all steps.
\State {\bf Parallel paths}: Use parallel paths of length $m-n$ to connect the nodes below cut $(u,n-1)$ to the nodes above cut $(\ell,n-1)$.
\State {\bf Disjoint paths}: Subtract $m-n$ from the sequence of Step 6. Now use Lemma~\ref{lemma:permutation} on this sequence to form the disjoint paths on the lower triangle.
\State {\bf Merging}: Complete the spanning tree by merging the two smallest subtrees in every upper triangle cut.
\end{algorithmic}
\end{algorithm}

\noindent {\bf Lower triangle:} {We first show that the cost of \minmin\ over the lower triangle is at most $4/3$ of the optimal tree over the lower triangle.}
\begin{lemma}\label{lemma:Lower_triangle}
The \minmin\ algorithm costs at most $4/3$ of the optimal over the cuts in the lower triangle. In other words, $$Alg^{(\ell,i)}\leq \frac{4}{3}OPT^{(\ell,i)},\quad i=1,2,\ldots,n-1,$$
where $Alg$ refers to the output of \minmin. Moreover, this implies the same approximation factor for the entire lower triangle energy, i.e., 
$Alg^\ell\leq \frac{4}{3}OPT^\ell$.
\end{lemma}
\begin{proof}
By the construction of Lemma~\ref{lemma:permutation}, the edges of the proposed tree over cut $(\ell,i)$ have $1,2,...,i$ successors (in some order). Therefore,
$Alg^{(\ell,i)}=\sum_{j=1}^i j^2={i(i+1)(2i+1)}/6$, for $i=1,\ldots,n-1.$
Comparing to the lower bound \eqref{eq:lb_lowercuts} for lower triangle cuts:
$$\frac{Alg^{(\ell,i)}}{OPT^{(\ell,i)}}\leq \frac{4i(i+1)(2i+1)}{6i(i+1)^2}=\frac{2(2i+1)}{3(i+1)}<\frac{4}{3}.$$
Since this is true for all cuts, it also holds for the entire lower triangle.
\end{proof}
\vspace{10pt}
\noindent {\bf Upper triangle:} {To analyze the \minmin\ algorithm over the upper triangle, we obtain a relation between the cost of the algorithm over successive cuts}:
\begin{lemma}\label{lemma:invariant}
For $i=1,\dots,n-2$, we have
\begin{equation}\label{eq:invariant}
Alg^{(u,i)}\leq Alg^{(u,n-1)}+2\sum_{j=i+1}^{n-1}\left[S_{u,j}+\frac{S_{u,j}^2}{j(j+1)}+j\right].
\end{equation}
\end{lemma}
\begin{proof}
Let $a_1,a_2,...,a_{i+1}$ be the number of successors for the edges of cut $(u,i)$, in non-decreasing order ($i\geq 2$). We know that $\sum_{j=1}^{i+1}a_j=S_{u,i}$. Since we merge $a_1,a_2$ in the higher level, the edges of cut $(u,i-1)$ will have $(a_1+a_2+1), (a_3+1),...,(a_{i+1}+1)$ successors. Therefore the cost of cut $(u,i-1)$ is
\begin{align*}
Alg^{(u,i-1)}&=(a_1+a_2+1)^2+ (a_3+1)^2+...+(a_{i+1}+1)^2\\
&=\sum_{j=1}^{i+1} a_j^2+2\sum_{j=1}^{i+1} a_j+2a_1a_2+i=Alg^{(u,i)}+2S_{u,i}+2a_1a_2+i.
\end{align*}
Since $a_1$ is the smallest number, it is upper-bounded by the average, i.e. $a_1\leq S_{u,i}/(i+1)$. Similarly, $a_2$ is the smallest among the rest, therefore $a_2\leq (S_{u,i}-a_1)/i\leq S_{u,i}/i$. So the algorithm satisfies:
$$Alg^{(u,i-1)}\leq Alg^{(u,i)}+2S_{u,i}+2\frac{S_{u,i}^2}{(i+1)i}+i.$$
By recursively applying this upper bound, we get \eqref{eq:invariant}.
\end{proof}
Next, add equation \eqref{eq:invariant} across all upper-triangle cuts and use the lower bound of Lemma~\ref{lemma:n4logn} to upper bound the energy of $Alg$ over the upper triangle:
\begin{lemma}\label{lemma:Upper_triangle}
The output of the \minmin\ satisfies: $\frac{Alg^u}{OPT^u}\leq 2+\bigO(\frac{1}{\log n})$.
\end{lemma}
\begin{proof}
We first sum \eqref{eq:invariant} over $i$ to get:
\begin{align}
Alg^u&=\sum_{i=1}^{n-1}Alg^{(u,i)}\\
&\leq (n-1)Alg^{(u,n-1)}+2\sum_{i=1}^{n-2}\sum_{j=i+1}^{n-1}\left[S_{u,j}+\frac{S_{u,j}^2}{j(j+1)}+j\right]\\
&= (n-1)Alg^{(u,n-1)}+2\sum_{j=2}^{n-1}(j-1)\left[S_{u,j}+\frac{S_{u,j}^2}{j(j+1)}+j\right]\label{eq:algu}
\end{align}
Now, for the first term in \eqref{eq:algu} we have:
$$(n-1)Alg^{(u,n-1)}=(n-1)\sum_{i=1}^{n}i^2=\frac{(n-1)n(n+1)(2n+1)}{6}.$$
For the $S_{u,j}$ term in the summation of \eqref{eq:algu} we have:
$$2\sum_{j=2}^{n-1}(j-1)S_{u,j}=2\sum_{j=2}^{n-1}(j-1)\left(n^2-\frac{j(j+1)}{2}\right)=\frac{3}{4}n^4-\frac{5}{2}n^3+\frac{9}{4}n^2-\frac{1}{2}n.$$
For the quadratic term in \eqref{eq:algu}, we have:
$$2\sum_{j=2}^{n-1}\frac{(j-1)S_{u,j}^2}{j(j+1)}\leq 2\sum_{j=2}^{n-1}\frac{S_{u,j}^2}{j+1}\leq 2\sum_{j=2}^{n-1}OPT^{(u,j)}\leq 2OPT^u,$$
where the second inequality is due to \eqref{eq:lb_uppercuts}.
Finally, we can calculate the last term of \eqref{eq:algu}, as:
$$2\sum_{j=2}^{n-1}(j-1)j=\frac{2}{3}n^3-2n^2+\frac{4}{3}n.$$
Replacing these polynomials back into \eqref{eq:algu} we get:
\begin{equation}
Alg^u\leq 2OPT^u+\frac{13}{12}n^4-\frac{5}{3}n^3-\frac{1}{12}n^2+\frac{2}{3}n.
\end{equation}
Dividing this by $OPT^u$ and using Lemma~\ref{lemma:n4logn} completes the proof.
\end{proof}

\vspace{10pt}
\noindent {\bf Overall approximation:}
{The worse approximation factor between the guarantees for $OPT^\ell$ and $OPT^u$ determines the overall result:}
\begin{theorem}\label{thm:minmin_approximation}
For a rectangular $n\times m$ $(n\leq m)$ grid with loads satisfying $d_i\in [d_{\min},d_{\max}]$ for all nodes $i\in V\backslash\{r\}$, the \minmin\ algorithm for the network reconfiguration problem with uniform resistances gives an approximation factor of $\alpha^2\left(2+\bigO(\frac{1}{\log n})\right)$, where $\alpha=d_{\max}/d_{\min}$. In particular, if the loads are uniform and as $n\rightarrow \infty$, the \minmin\ algorithm gives a $2$-approximation. 
\end{theorem}

\begin{proof}
For square grids with uniform loads, the approximation result follows immediately from Lemmas~\ref{lemma:Lower_triangle} and \ref{lemma:Upper_triangle}. Moreover, for the rectangular grid with uniform loads, we can apply the \minmin\ algorithm on the lower and upper triangle parts, and use the middle section to connect the two triangles simply by parallel disjoint paths. The analysis is then similar to the square case. 

For non-uniform loads, we consider the uniform counterpart of this instance, which is the same graph with $d_i=d_{\min}$ for all $i\in V\backslash\{r\}$. Running the \textsc{\textsc{Min-min}} algorithm on this uniform case outputs a tree $T$ whose loss is at most twice of the optimal tree in the uniform setting (call the optimal tree $T^*_u$, and its loss $OPT_u$). Let $f$ and $\tilde{f}$ be the electrical flows on tree $T$ with the actual and modified loads, respectively; then we have $\tilde{f}\leq f\leq \alpha \tilde{f}$. This gives the following inequality regarding energies:
$\E(f)\leq \alpha^2\E(\tilde{f})\leq \alpha^2\big(2+\bigO(\frac{1}{\log n})\big)OPT_u$.

It only remains to argue that $OPT_u\leq OPT$, where $OPT$ is the loss of the optimal tree (call it $T^*$) in the original instance. This is true because if we reduce the loads on $T^*$ to $d_{\min}$, we decrease its loss, but on the other hand, the resulting energy should still be more than $OPT_u$, by the optimality assumption of $OPT_u$.
\end{proof}

Note that $\alpha$ in the above approximation captures the ratio between the biggest and the smallest loads, and is usually independent of $n$ in practice (typically, loads do not vary a lot in a realistic scenario). The above approximation ratio can be thought of as $\min \{2\alpha^2,n\}$, where $n$ comes from the general result of Theorem~\ref{thm:spt} in the case of a large $\alpha$.

\section{A Generalization of {\sc Min-Min} and Computational Results} \label{sec:lm}




\subsection{The Layered-Matching Heuristic}
We extend our intuition from the theoretical results and propose a generalization of {\sc Min-Min}. The main idea is to partition the graph into layers and connect each layer to the upper layer in a balanced way, similar to {\sc Min-Min}. To partition the graph into layers, we can use an arbitrary breadth first search tree rooted at $r$, which results in layers based on the hop-distance from the root (this is like the diagonal cuts on the grid), starting with layer $L_1$ comprised of nodes adjacent to the root. To connect layers $L_1, \hdots, L_k$ in a balanced manner, we use the flow relaxation, and find the best matching that creates a flow which is close to the relaxed solution in terms of $L_{\infty}$ norm. In particular, to connect nodes in layer $L_k$ to the upper layer $L_{k-1}$, using $f^r$ the solution of the flow relaxation, $d_i$ be the demand of node $i$, let $x_{ij}$ be the indicator of picking an edge $(i,j)\in E$. Then, finding the best matching reduces to solving the following integer program:
\begin{equation}\label{eq:matching_IP}
    \begin{aligned}
    \min &~~~ \epsilon\\
    \text{s.t.} &~ \sum_{\substack{i \in L_{k-1}:\\ (i,j) \in E}} x_{ij} = 1 && \forall \; j \in
    L_k,\\
    &~ - \epsilon \leq x_{ij} d_j-  f^{r}_{ij} x_{ij} \leq \epsilon, && \forall \; i \in
    L_{k-1}, j \in  L_{k} : (i,j) \in E,\\
        &~  x_{ij} \in \{0,1\}, && \forall \; i \in
    L_{k-1}, j \in  L_{k} : (i,j) \in E. 
\end{aligned}
\end{equation}
This IP can be solved very fast in practice since it finds a (local) matching between two layers of nodes (as discussed in the next section). Once we find the matching between layers $k$ and $k-1$, we contract each node of $L_{k-1}$ with its successors, while replacing the demand of that node with the total demand of the corresponding subgraph. By repeating this process in a bottom-up fashion, all the nodes get connected to the root via a single path, hence we obtain a valid spanning tree. The full description of the heuristic is in Algorithm \ref{alg:heuristic}.


\begin{algorithm}[t]\label{LM} 
\caption{Layered-Matching (LM) Heuristic}
\label{alg:heuristic}
\begin{algorithmic}[1]
\State Compute BFS tree rooted at the root $r$.
\State Let $L_0,\dots,L_t$ be the layers formed in the BFS tree, where $L_0=\{r\}$.
\For{$k=t,t-1,...,1$}
\State Match each node of $L_k$ to exactly one node in $L_{k-1}$ via solving IP \eqref{eq:matching_IP}.
\For{$j\in L_{k-1}$}
\State Contract subgraph rooted at $j$ into a single node with demand equal its entire subgraph.
\EndFor
\EndFor
\end{algorithmic}
\end{algorithm}

\subsection{Computational Results for $25\times 25$ Sparsified Grids} As previously discussed, electricity distribution networks resemble sparsified grids. For computational experiments, we constructed {25} instances {for each sparsification probability} on $25 \times 25$ grids where edges are deleted independently with some probability $p \in \{0.05,0.1, 0.2\}$ as long as they do not disconnect the (current) graph; an example of such sparsified grids is given in Figure~\ref{fig:realnetwork}. Demands for each vertex are sampled uniformly randomly in $[0.5, 1.5]$, resistances are also sampled uniformly randomly in $[1, 10]$ and the root node is placed in the corner. To incorporate the acyclic support constraint, we utilized Martin's \cite{martin1991} extended formulation for spanning trees; we refer the reader to Appendix \ref{computaions-appendix} for more details on the MIP formulation and computational plots.\footnote{The code used for the simulations in this paper can be found here: \url{https://github.com/hassanmortagy/Electrical-Flows-over-Spanning-Trees}}

We benchmarked\footnote{We implemented all algorithms in \texttt{Python} 3.7, utilizing \texttt{numpy} and \texttt{networkx} for some of our functions. We used these packages from the Anaconda 4.7.12 distribution, with \texttt{Gurobi} 9 \cite{Gurobi2020} as a solver for the MIP.} the performance of depth-first search (DFS) trees, \textcolor{blue}{shortest-path trees (SPT)}, \textsc{RIDe}, the Layered-Matching (LM) heuristic, the branch exchange heuristic and the convex integer program. The branch exchange variant we utilized starts a new iteration once an improving solution is found, as opposed to looking for the exchange that results in the most improvement. We also considered a variant that uses binary search on the value of the improvement where we only do an exchange if it results in an improvement of value at least $T$; if no such improvement exists we divide $T$ by two and proceed. We found that the latter version was significantly slower.

\begin{figure}
    \centering
        \includegraphics[scale = 0.46]{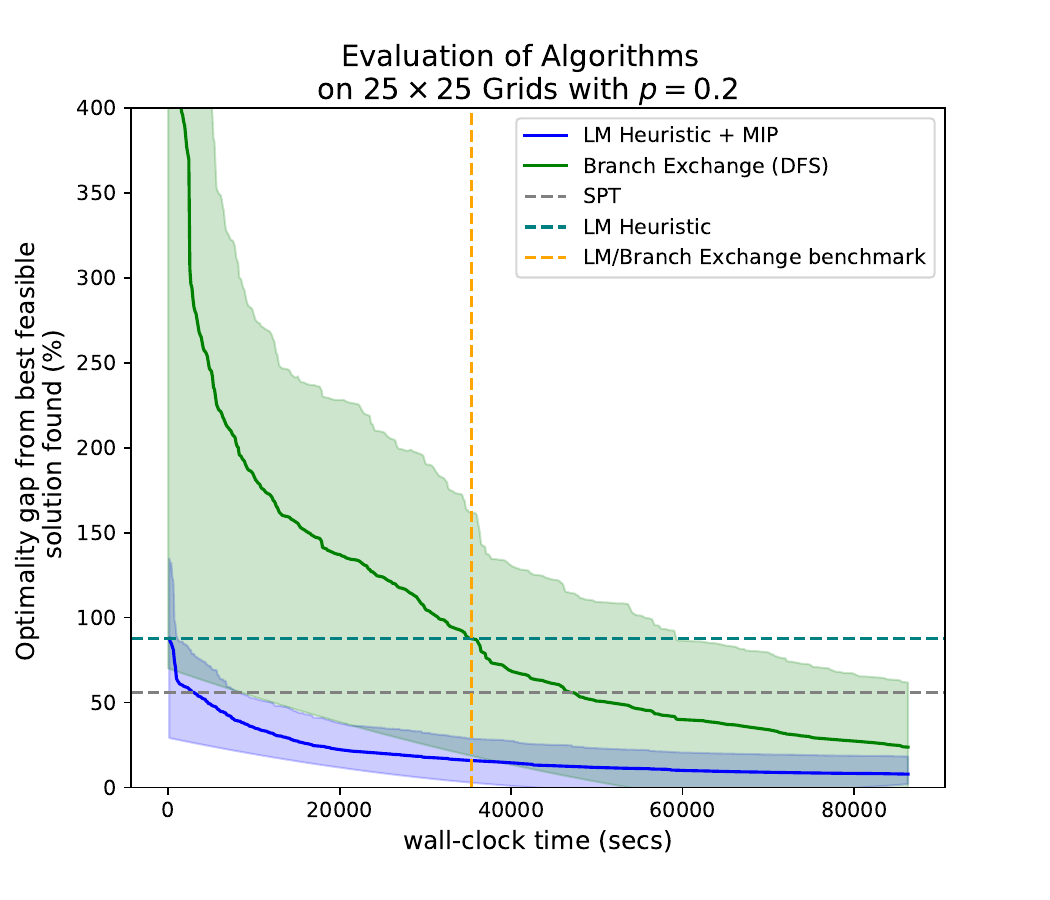}
        \includegraphics[scale = 0.46]{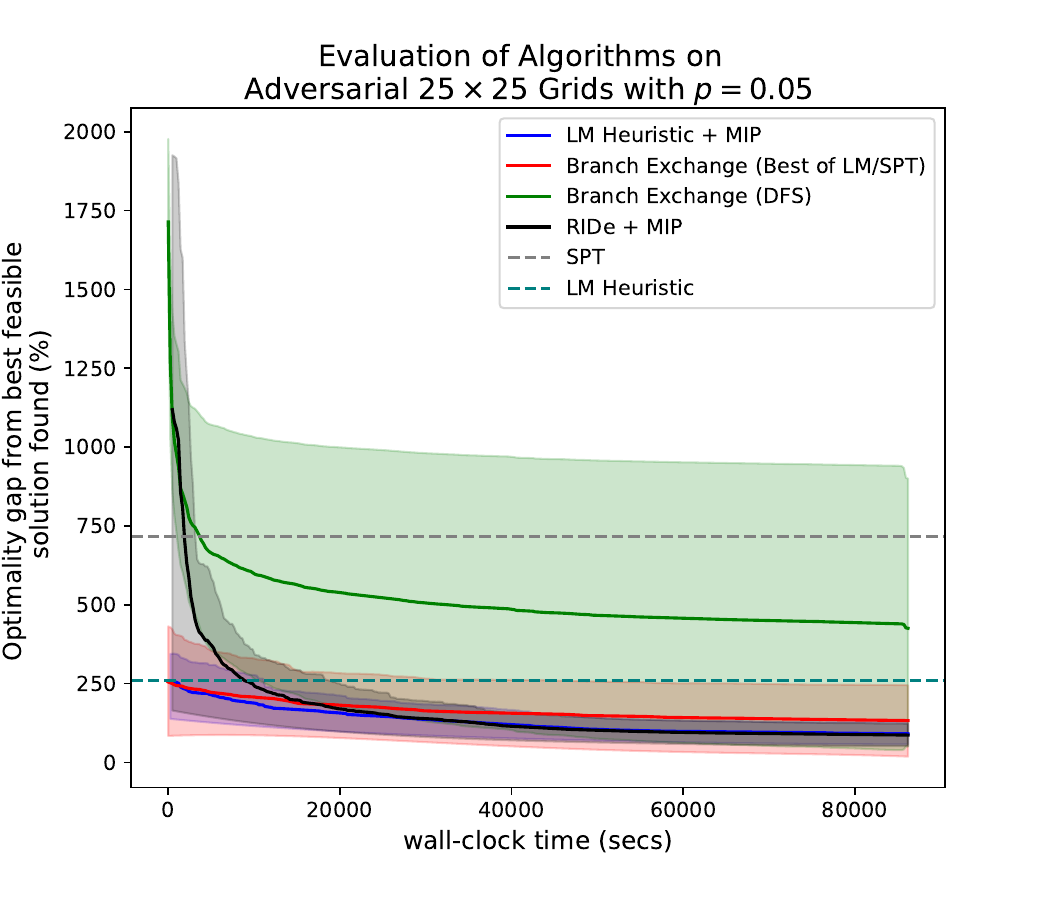}
    \caption{\textbf{Left:} Plots comparing performance of branch exchange (initialized with DFS trees) and the mixed integer program initialized with the layered matching heuristic solution, on sparse $25 \times 25$ grids with sparsification probability $p=0.2$. The demands and resistances are randomly chosen in $[0.5,1.5]$ and $[1,10]$ respectively. The dotted horizontal line compares the quality of the Layered Matching heuristic solution and the color margins represent confidence intervals across the different instances. {\textbf{Right:} Benchmarking performance of all algorithms on instances where resistances are chosen \emph{adversarially} as explained in Appendix~\ref{app:SPT}. Further, the demands are also randomly chosen in $[0.5,1.5]$ and the sparsification probability $p=0.05$.}}
    \label{fig:comp}
    \vspace{-13 pt}
\end{figure}

Our computations show that the algorithms proposed in this paper are orders of magnitude faster than branch exchange (initialized with a random DFS tree) while having comparable performance. For example, on $25 \times 25$ grid instances with sparsification probability $p = 0.2$, the mean \textcolor{blue}{time} taken by the LM heuristic is 1.35 seconds, whereas the mean time it takes the branch exchange heuristic to attain the same cost as LM is around 35,000 seconds (see Table~\ref{tab:comp-25-table-alg}). We believe that this \textcolor{blue}{significant improvement} will allow system operators to minimize losses even on an hourly basis as demand patterns change.

Moreover we find that \texttt{Gurobi} mostly failed to even find a feasible solution for sparsified $25\times 25$ grids instances in a 24-hour time limit. To help the IP solver, we provide a warm start solution using different algorithms and let the solver run for 24 hours. This outperforms the branch exchange heuristic after running for the same one day time limit (see Figure \ref{fig:comp}-(left)). In addition, running the MIP with the LM output as a warm-start obtains the same performance as initializing branch exchange with the LM output (as opposed to a DFS tree), however the MIP additionally gives provable optimality gaps. \textcolor{blue}{We report the gap in solution quality with respect to the best feasible solution found after running all algorithms for 24 hours. In particular,  we found that the best feasible solution was obtained by running the MIP with the LM output as a warm-start in 46.7\% of the instances, and branch exchange initialized with the LM output in 53.3\% of the instances.\footnote{For these solutions, the provable gaps using MIPs were around 15-30\%.}}

\begin{table}[t]
\centering
{\renewcommand{\arraystretch}{1.3}
\begin{tabular}{|p{0.7cm}||p{0.5cm}p{0.5cm}|p{0.5cm}p{0.5cm}|p{0.5cm}p{0.5cm}|p{0.5cm}p{0.5cm}|c|} \hline
			\multirow{2}{*}{ \small $p$}    & \multicolumn{2}{c|}{\small DFS}  & \multicolumn{2}{c|}{\small {SPT}}           & \multicolumn{2}{c|}{\small \textsc{RIDe}}                           & \multicolumn{2}{c|}{\small \textsc{LM}} & \multirow{2}{2.25 cm}{\small{Branch Ex. time to attain LM} gap}\\
			     & \multicolumn{1}{c}{ \small Gap} & \multicolumn{1}{c|}{\small Time} &
			       \multicolumn{1}{c}{ \small {Gap}} & \multicolumn{1}{c|}{\small {Time}}& \multicolumn{1}{c}{\small Gap} & \multicolumn{1}{c|}{\small Time} & \multicolumn{1}{c}{\small Gap} & \multicolumn{1}{c|}{\small Time} & \\ \hline \hline
$0.05$ &  60.80 & 0.57 &  \textcolor{blue}{0.56} & {\textcolor{blue}{0.14}} & 8.13&196.8& 1.22 & {\bf 2.93}  &  {\bf 5680.30}  \\ \hline
  $0.1$ & 49.69 &  0.35&  \textcolor{blue}{0.58} & { \textcolor{blue}{ 0.12}} & 6.36& 125.9& 1.12 & {\bf 1.92} & {\bf 9380.28} \\ \hline
  $0.2$ &  39.43 &  0.21&  \textcolor{blue}{0.56} & { \textcolor{blue}{0.10}} & 5.73& 77.07& 0.90& {\bf 1.35} &  {\bf 35181.11}\\ \hline \hline			  $\textcolor{blue}{0.05^*}$ & \textcolor{blue}{17.13} & \textcolor{blue}{0.54} &  {\bf\textcolor{blue}{7.15}} &  \textcolor{blue}{0.10} &\textcolor{blue}{10.92}& \textcolor{blue}{134.9}& \textcolor{blue}{\bf{0.80}}& {\bf \textcolor{blue}{0.73}} &  {\bf \textcolor{blue}{--}}\\ \hline		
		\end{tabular}}
\caption{Comparing the mean optimality gap from best feasible solution found ($\times $100\%) and mean running time (secs) of different algorithms on sparse $25 \times 25$ grids with varying probability of sparsification ($p$). \textcolor{blue}{The last column denotes the time taken by branch exchange to reach the same solution quality as LM heuristic. The first three rows provide our results on the first set of instances. The last row provides the results for the second set of adversarial instances, where branch exchange did not attain the same cost as LM within 24 hours.}}
    \label{tab:comp-25-table-alg}
\end{table}

\textcolor{blue}{For the set of random instances considered, SPTs outperformed LM and \textsc{Ride} (see Table \ref{tab:comp-25-table-alg}). However, since SPTs are oblivious to the node demands this behavior is not generalizable. To explore this further, we  constructed sparsified grid instances with edge resistances generated using adversarial distributions (mimicking the worst-case instances for SPTs, given by a Hamiltonian path of low resistance). In these instances, we find that the performance of SPTs is arbitrarily bad and the ratio of the cost between SPTs and LM grows linearly in the dimension of the graph (see Appendix \ref{app:SPT} for more details). Moreover, we find that LM outperformed all the other methods, as shown in Figure \ref{fig:comp}-(right). Since LM is robust to the structure of the graph, edge resistances and node demands and SPTs are very cheap to compute, we propose taking the best of both for initializing the MIPs.} 

Finally, the computations suggest that the approximation factor obtained by \textsc{RIDe} in practice is much better than the worst-case theoretical bounds we show. Further, the performance of \textsc{RIDe} (as expected) and LM significantly improve as the graph gets sparser (see Table \ref{tab:comp-25-table-alg}), and their improvement over branch exchange becomes more pronounced. This is very desirable since electricity distribution networks are indeed typically sparse in practice.

\section{Conclusions and Open Questions}\label{sec:conclusion}\vspace{-5pt}
In this paper we studied the network reconfiguration problem from power systems (for distribution networks) through the lens of approximation algorithms. We provided approximation algorithms for different scenarios with restrictions on graph structure, line resistances, and node demands.
\textcolor{blue}{Some open questions} still remain, including the
extension of the $\sqrt{n}$-approximation (or even constant-factor approximation) to planar graphs,
analysis for the iterative deletion of the min-flow edge (introduced by Shirmohammadi and Hong~\cite{shirmohammadi1989reconfiguration}),
and the hardness of the problem for grids or planar graphs.
\subsection*{Acknowledgements}
The authors would like to thank the participants of the Real-Time Decision Making Reunion Workshop,  Mixed Integer Programming Workshop, and the IEEE Power \& Energy Systems General Meeting for valuable feedback. In particular, we would like to thank David Williamson for pointing us to low-stretch spanning trees, Alexandra Kolla for pointing us to spectrally-thin trees, \textcolor{blue}{Tasos Sidiropoulos for pointing us to the grid embedding of planar graphs, and the anonymous referees for numerous useful suggestions to improve this manuscript.} A part of this work was done while the authors were visiting the Simons Institute, UC Berkeley. We gratefully acknowledge the financial support from NSF grants CRII-1850182, CCF-1733832, CCF-1331863, and CCF-1350823.

\bibliographystyle{plain}
\bibliography{ref.bib}

\appendix

\newpage

\section{Detailed Overview of Existing Techniques} \label{sec:relwork}
\paragraph{Related work in power systems:} The problem of reconfiguring the electric distribution network to minimize line losses was first introduced by Civanlar \etal \cite{civanlar1988distribution} and Baran and Wu \cite{baran1989network} where they introduced and implemented an algorithm called ``Branch Exchange'', which tries to locally improve the objective by swapping two edges of the graph. Unlike branch exchange that maintains a feasible spanning tree during its execution, there are other algorithms that start with the entire graph and delete edges one by one until a feasible solution is obtained \cite{shirmohammadi1989reconfiguration}. We discussed this approach further in Section~\ref{sec:flow_relax}. Subsequently, many other heuristic algorithms were proposed for the reconfiguration problem, including but not limited to genetic algorithms \cite{enacheanu2008radial}, particle swarm optimization \cite{kumar2014power}, artificial neural networks \cite{salazar2006artificial}, etc.
The missing part in all these heuristics is a rigorous theoretical performance guarantee that shows why/when these algorithms perform well. To that end, Khodabakhsh \etal \cite{khodabakhsh2018submodular} recently showed that the reconfiguration problem is equivalent to a supermodular minimization problem under a matroid constraint. 

\paragraph{Existing techniques in combinatorial optimization:} One may ask the question if a simple spanning tree like breadth-first or depth-first search tree will be a good solution to the reconfiguration problem. It is shown in \cite{khodabakhsh2018submodular} that if the edges are identical, the optimal tree will include all the edges incident to the root. So the BFS tree might be a better candidate. However, as we showed in Section~\ref{sec:n_and_logn}, a BFS tree can have a loss of $\Omega(n)$ times the optimal loss. Depth-first search can be worse; for example, a cycle with spokes and root in the center has a gap of $\Omega(n^2)$. (the optimal tree is a star with linear cost, while the DFS tree will take the cycle with cubic cost.)

Our problem looks like a {\it bicriteria tree} approximation; on the one hand, we want to connect demands to the root via shortest paths, on the other hand, we want the paths to be disjoint, i.e., degree of the nodes (except the root) in the tree must be small. Similar problems have been studied in the Computer Science literature: K{\"o}nemann and Ravi \cite{konemann2002matter} consider finding a minimum cost spanning tree subject to maximum degree at most $B$. 
The variation where each node has its own specified degree bound is also studied by Fekete \etal \cite{fekete1997network} and Singh and Lau \cite{singh2007approximating}. Khuller \etal \cite{khuller1995balancing} also define Light Approximate Shortest-path Trees (LAST), in which a tree is $(\alpha,\beta)$-LAST if the distances to the root are increased by at most a factor of $\alpha$ (compared to the original graph), while the cost of the tree is at most $\beta$ times the minimum spanning tree. However, the main difficulty 
in using these approximations for the reconfiguration problem is accounting for the resultant flow in the spanning trees. The cost on the edges of the tree are then no longer linear (as in the above mentioned results) or even quadratic. 

Electrical energy is minimized when considering {\it edge-disjoint paths} to connect nodes  \cite{seymour1980disjoint}. However, the complete disjointness is rarely achievable in our problem (unless the graph is a star with root in the center), and hence we want to limit the number of flows that merge together. This is more related to the Edge-Disjoint Path with Congestion (EDPwC) problem, studied by
Andrews \etal \cite{andrews2010inapproximability}, which is as follows: Given an undirected graph with $V$ nodes, a set of terminal pairs and an integer $c$, the objective is to route as many terminal pairs as possible, subject to the constraint that at most $c$ demands can be routed through any edge in the graph. They show hardness of approximation for EDPwC problem. The main differences with our problem are that in EDPwC there is a hard constraint on the flow routed through any edge, while in our problem there is a quadratic cost associated with that flow, as well as the additional spanning tree constraint in the reconfiguration problem.

A natural question is if there are some known graph families where one can exploit existing structures to find an approximation. We consider planar graphs, also motivated by the application since many distribution networks are designed that way. As we saw in Section~\ref{sec:n_and_logn}, one way to obtain useful lower bounds for the objective function is via generating a packing of cuts that are small in size, i.e., do not have many edges. A celebrated result in the theory of planar graphs is the existence of a small set of vertices, called the {\bf vertex-separator}, that can disconnect the graph into components of almost equal size  \cite{lipton1979separator}. This can even be applied in a recursive manner, as shown by Frederickson \cite{federickson1987fast}, to divide the graph into $\bigO(n/r)$ regions with no more than $r$ vertices each, and $\bigO(n/\sqrt{r})$ boundary vertices in total. However, it is unclear how to bound the cost of these cuts or regions in the reconfiguration problem, unless we have more information about the direction of the resultant flow on boundary edges.

\paragraph{Existing techniques in electrical flows:} Electrical flows have been an active area of research in the past two decades, due to their computational efficiency and numerous applications to graph theory problems. In particular, it was shown that one can compute an electrical flow in a graph in near-linear time \cite{peng_2011,kelner_2013,cohen_2019}. Moreover, various novel graph algorithms involve computing an electrical flow as a subroutine. For example, Madry \cite{Madry_2016} and Christiano \etal \cite{Madry_2010} use electrical flows to obtain the fastest algorithms for the \textsc{Max-flow} problem so far.

If we relax the spanning tree constraint, our problem reduces down to the standard problem of computing an electrical flow in a graph that satisfies the demands.  Even though the support of an electrical flow in its full generality does not form a spanning tree, still, it maybe be beneficial to use the minimum energy electrical flow in the graph as a starting point. Shirmohammadi and Hong \cite{shirmohammadi1989reconfiguration} follow this approach and propose an iterative algorithm for the reconfiguration problem, where in each iteration they compute the electrical flow and delete the edge with smallest flow such that the graph remains connected. They demonstrated experimentally that their iterative algorithm performs well in practice but they provide no theoretical guarantees. In Section \ref{sec:flow_relax}, we proposed a similar iterative edge-deletion algorithm and prove its approximation bound in Theorem \ref{rand theorem}.

At the heart of these edge deletions is the question of whether we can delete edges from the graph without increasing the energy cost too much. One approach that can be used to address this question is {\it spectral sparsification} \cite{spielman_2008}, which aims to reduce the number of edges in the graph while maintaining $(1 \pm \epsilon)$ approximations of the Laplacian quadratic form. In their classic result, Spielman and Srivastava \cite{spielman_spars} show that one can construct such a sparsifier with $\Tilde{\bigO}(n/\epsilon^2)$ edges. Chu \etal \cite{peng_2018} slightly improve upon the results of Spielman and Srivastava, bringing the number of edges down to $\Tilde{\bigO}(n/\epsilon)$ for some specific instances. 
Using the fact that electric flows are fully characterized by the Laplacian quadratic form, one may conclude that by using such a sparsifier we can reduce the number of edges in the graph without significantly increasing the energy cost. However, this is not true, because to obtain such sparsifier they compensate the deletion of edges by changing the weights (i.e. resistances) on the edges. Thus, since we assume resistances are fixed, to the best of our knowledge, the existing spectral sparsification approach does not extend to our problem. This motivates the need of a novel approach to handle edge deletions without increasing the energy too much.

\paragraph{Uniform Spanning Trees:} To deal with the iterative edge deletions, we consider sampling from  distributions over spanning trees. Random spanning trees are one of the most well-studied probabilistic combinatorial structures in graphs. Recent work has specifically considered \emph{product distributions} over spanning trees where the probability of each tree is proportional to the product of its edge weights. (This is motivated by the desirable properties and numerous applications of such product distributions.) For example, Asadpour \etal \cite{Goemans_2010} break the $\bigO(\log n)$ barrier of the \textsc{ATSP} problem by rounding a point in the relative interior of the spanning tree polytope by sampling from a maximum entropy distribution over spanning trees; this maximum entropy distribution turns out to be a product distribution. Moreover, a beautiful property of product distributions over spanning trees is the 
fact that the marginal probability of an edge being in a random spanning tree is exactly equal to the product of the edge weight and the effective resistance of the edge (see Appendix \ref{sec:Uniform Spanning Trees}). This is a fact that we exploit in our \textsc{RIDe} algorithm.



\paragraph{Low-stretch trees:} Finally, another relevant approach in the electric flows literature, entails {\it low stretch trees}. Given a weighted graph $G$, a low-stretch spanning tree $T$ is a spanning tree with the additional property that it approximates distances between the endpoints of any edge in $G$. 
In particular,\footnote{We follow the definition of Elkin \etal \cite{spielman_2004}, but this definition slightly differs in the denominator from others given in the literature. Abraham and Neiman \cite{ofer_2012} and Abraham {\etal} \cite{ofer_2008} define the stretch as $ \frac{d_T (u,v)}{d_G(u,v)}$, where $d_G$ is the shortest-path metric on $G$ with respect to the edge weights. Note that these definitions are equivalent if the edge weights are uniform.} the \emph {stretch} of an edge $e =(u,v)$ is the ratio of the (unique) shortest path distance between $u$ and $v$ in $T$ to $r_e$ (the weight of edge $e$ in $G$). Furthermore, the \emph {total stretch} of $T$ is defined as the sum of the stretch of all edges in $G$. Kelner \etal \cite{kelner_2013} show that for any tree, the gap between the energy of the flow in that tree and the flow in the original graph, is at most the total stretch of that tree. Naturally, one may wonder if there exists a low value for the (total) stretch such that all graphs have a spanning tree with that stretch. The answer to that question is unfortunately no.  Abraham and Neiman \cite{ofer_2012} show that one can construct a spanning tree $T$ for any connected graph with total stretch at most $\bigO(m \log n \log \log n)$ in near-linear time (Theorem 2.11 in \cite{kelner_2013}); this bound is tight up to an $\bigO(\log \log n)$ factor because Alon et. al \cite{karp_1995} show that the total stretch is $\Omega(m\log n)$ for certain graph instances. Thus, this implies that the energy cost of $T$ is at most $\tilde{\bigO}(m)$ times that of the original graph. We improve upon this approximation result using our \textsc{RIDe} algorithm.



\section{Missing background Information} \label{sec:background}
We give a review of preliminaries on electrical flows, graph Laplacians and their pseudoinverse, and matrix inversion results. We refer the reader to \cite{williamson2019network,lyons_2017} for more details. 

\subsection{The Graph Laplacian} \label{sec:graph lap}
Let $G = (V,E)$ be a connected and undirected graph with $\abs{V}=n$, $\abs{E}=m$. Each edge $e\in E$ is also associated with a \emph{resistance} $r_e>0$. The inverse of the resistance is called \emph{conductance}, defined by $c_e=1/r_e$. 
Let $B \in \mathbb{R}^{n \times m}$ be the vertex-edge incidence matrix upon orienting each edge in $E$ arbitrarily. Also, let $R$ be an $m\times m$ diagonal resistance matrix where $R_{e,e} = r_{e}$. We define the weighted Laplacian $L := BCB^T$, where $C = R^{-1}$. Since $C$ is a positive definite and symmetric matrix, we could write  $ L = (C^{1/2}B^T)^T (C^{1/2}B^T)$, which implies that $L$ is positive-semi definite, since $x^TLx \geq 0$ for all $x \in \mathbb{R}^n$.

Let $\mathbf{1}$ be the all-ones vector. For any matrix $A \in \mathbb{R}^{m \times n}$, denote the span of the columns of $A$ by $\mathrm{im}(A) \subseteq \mathbb{R}^{m}$. It is well known that if $G$ is connected, the only vector in the nullspace of the Laplacian $L$ is the all-ones vector $\mathbf{1}$. In what follows, we will use the Moore-Penrose pseudoinverse, denoted by $L^\dagger$, to invert the Laplacian.
Since $L$ is symmetric and positive semi-definite, we can write $L$ in terms of its eigen-decomposition $L= \sum_{i=1}^n \lambda_i u_i u_i^T$, where $0 = \lambda_1  \leq \dots \leq \lambda_n$ are the eigenvalues of $L$ sorted in increasing order and $u_i$ are the corresponding singular orthonormal vectors. Now, the pseudoinverse could be conveniently characterized using $L^\dagger = \sum_{i=2}^n \frac{1}{\lambda_i}u_i u_i^T$. Observe that $L L^\dagger = \sum_{i=2}^n u_i u_i^T$ and is thus a projection matrix that projects onto $\mathrm{im}(L)$. In other words, for any vector $x \in \mathbb{R}^{n}$ such that $x^T \mathbf{1} = 0$, $L L^\dagger x = x$.

\subsection{An Introduction to Electrical Flows} \label{intro to elec flows}
Given a graph $G = (V,E)$, a root $r \in V$ and demands $d_i \geq 0$ for all $i \in V \setminus\{r\}$, we begin by assigning a demand $d_r = - \sum_{i \in V \setminus\{r\}} b_i$ to the root, and we collect these demands into a demand vector $b \in \mathbb{R}^{n}$. An electrical flow is a feasible flow that satisfies demands, while also minimizing the electrical energy. Hence, computing an electrical flow amounts to solving the following problem:
\begin{equation}\label{energy}
    \begin{aligned}
     \min &\quad \E(f) = f^TRf\\
    \text{s.t.} &\quad Bf = b
\end{aligned} 
\tag{P2}
\end{equation}
The optimality conditions\footnote{Ohm's Law says that the electrical flow on any edge is equal to the potential difference divided by the
resistance of the edge (or equivalently multiplied by the conductance).} of \eqref{energy} (i.e. the problem of computing an electrical flow) imply the existence of a vector of potentials on the nodes (dual variables) $\phi  \in \mathbb{R}^n$ such that
\begin{equation} \label{ohm}
    f_{u,v}^* =  \frac{\phi_v - \phi_u}{r_{u,v}}
\end{equation}
or $f^* = CB^T \phi$ in matrix notation. By pre-multiplying this equation with $B$ on both sides, we have $\phi = L^\dagger b$, which is well-defined since $\mathbf{1}^T b = 0$. Using these facts, one can easily show that
\begin{equation}\label{energy-potentials}
\E(f^*) = R^Tf^*R =  \phi^T L \phi = b^T \phi = b^T L^\dagger b.
\end{equation}

For any pair of vertices $u$ and $v$, let  $\chi_{uv} \in \mathbb{R}^{n}$ be a vector with a $-1$ in the coordinate corresponding to $u$, a 1 in the coordinate corresponding to $v$, and all other coordinates equal to 0. The \emph{effective} resistance between a pair of vertices $u,v$ is defined as
 \begin{equation}\label{effective resis}
     \Reff(u,v) := \chi_{uv}^T L^\dagger\chi_{uv}.
 \end{equation}
In other words, it is the energy of sending \emph{one} unit of electrical flow from $u$ to $v$. 

For any edge $e = (u,v)$, it is well known that $\Reff(e) \leq r_e$, where equality holds if and only if $e$ forms the only path between $u$ and $v$ (see for example Theorem D in \cite{Klein_1993}). Intuitively, if $\Reff(e) = r_e$, then, upon sending one unit of electrical flow between the endpoints of $e$, all that flow goes through $e$, which implies that $e$ is \emph{bridge} since otherwise we could otherwise reroute some of the flow through another $u-v$ path and decrease the effective resistance, which would contradict the optimality of the electrical flow.

\subsection{Uniform Spanning Trees} \label{sec:Uniform Spanning Trees}
In a weighted graph, a uniform distribution of spanning trees is one such that probability of each tree is proportional to the product of the weight of its edges. 
\begin{definition}
For $w : E \to \mathbb{R}_{++}$, we say $\lambda$ is a $w-$uniform spanning tree distribution if it is a product distribution and for any spanning tree $T \in \mathcal{T}$
$$\mathbb{P}[T] \propto{\prod_{e \in T}w(e)}.$$
\end{definition}
Let $\lambda_e := \mathbb{P}_{T \sim \lambda} (e \in T)$ be the marginal probability of an edge $e \in E$. It is known that (see for example \cite{Durfee_2016,lyons_2017,williamson2019network})
$$  \lambda_e = w(e) \chi_e^T L_w^\dagger\chi_e  \qquad \text{and} \qquad \sum_{e \in E} \lambda_e = n-1,$$
where $L_w$ is the weighted Laplacian defined with respect to the weights $w$. In particular, the vector of marginal probabilities $\lambda_e$, $e \in E$, is in the spanning tree polytope. In this work, we specifically consider the case when we choose $w$ to be the conductances, i.e. $\lambda$ is a $c-$uniform distribution. Hence, using \eqref{effective resis} we know that $\lambda_e = c_e \chi_e^T L^\dagger\chi_e = c_e \Reff(e)$, where $L$ is the weighted Laplacian defined with respect to the conductances (see Appendix \ref{sec:graph lap}). Therefore, for any spanning tree $T \in \mathcal{T}$,
\begin{equation} \label{prob_dist}
    \mathbb{P}[T] = \frac{\prod_{e \in T}c_e}{K} \qquad \text{and} \qquad \sum_{e \in E} c_e \Reff(e) = n-1,
\end{equation}
where $K = \sum_{T \in \mathcal{T}}{\prod_{e \in T}c_e}$ is the normalization factor. 

Observe that under a $c-$uniform spanning tree distribution, if an edge $e = (u,v)$ has a low marginal probability $c_e \Reff(e)$, then there are relatively a lot of paths between $u$ and $v$ excluding $e$. Therefore upon deleting an edge $e$ from the graph, it would not be costly to reroute the flow going through edge $e$. Similarly, if an edge has a high marginal probability, then rerouting the flow upon deleting that edge would be relatively very costly. This is the crucial observation that we use in our \textsc{RIDe} algorithm.

\subsection{A Note on Reactive Power}\label{app:reactive}
In this paper, we assumed that the demands ($d_i$'s) are real-valued parameters. However, in energy systems, demands are usually complex numbers $d=p+\mathbf{i}q$, capturing the active $(p)$ and reactive $(q)$ parts of the demand. Consequently, the loss on each line will be $r_e[(\sum_{i\in \text{succ}(e)}p_i)^2+(\sum_{i\in \text{succ}(e)}q_i)^2]$. Note that in this case, the objective function can be decomposed into two additive parts, in which one is only a function of real demands ($p$), and the other is only a function of the reactive part ($q$). We argue that our results would still hold. In particular, the approximate solutions in Theorems~\ref{rand theorem},\ref{thm:spt},\ref{thm:general_cuts} are independent of the demands; hence, the approximation factor would hold for both the active and reactive parts of the objective function. In Theorem~\ref{thm:minmin_approximation}, the \minmin\ algorithm would output the same spanning tree if performed with either $p$ or $q$, given the uniform assumption on (complex) loads; therefore, the approximation factor holds for both parts of the objective function. 

\section{Convex Optimization over the Flow Polytope}\label{sec:convex}
One can think of the reconfiguration problem \eqref{eq:p0} as minimizing a convex function over the vertices of a (flow) polytope. However, the general results from convex optimization over flow polytopes do not lead to good guarantees. We propose a randomized algorithm \textsc{RIDe} that rounds the fractional solution obtained from the flow relaxation to vertices of the flow polytope, while providing an $\bigO(m-n)$ approximation guarantee. But before that, here we review this new perspective on our reconfiguration problem, and some main results of interest in convex optimization over polytopes.

For a directed and connected graph $G$ with vertex-edge incidence matrix $B$ and a demand vector $b:V \to \mathbb{R}$, the general flow polytope is given by $P = \{f\in \mathbb{R}^m:\, Bf =b, f \geq 0\}$. It is known that the support of the vertices of the flow polytope, denoted by $\text{vert}(P)$, forms a spanning tree (see Theorem 7.4 in \cite{bertsimas_97}). This follows from the fact that there is a one-to-one correspondence between the bases of the graphic matroid defined by $G$ and the linear matroid defined on the incidence matrix $B$. In other words, there is a one-to-one correspondence between the spanning trees of $G$ and subsets of $n-1$ linearly independent columns of $B$ (note that the rank of $B$ is $n-1$). Therefore, a basic solution of the flow polytope will be one in which the flow is sent along a spanning tree (when ignoring the edge directions). If such a basic solution (spanning tree) additionally satisfies the flow conservation constraints while taking the edge directions into account, we obtain a basic \emph{feasible} solution or a vertex of $P$. Since for undirected graphs we can replace each edge with two directed edges, we can replace the flow conservation constraints of the network reconfiguration problem given in \eqref{flow} using the constraints given by $P$ above, where an extreme point of that polytope will then correspond to a flow sent along a spanning tree rooted at the root $r$. Such an extreme point is precisely one of the feasible solutions of \eqref{eq:p0} and \eqref{eq:p1}. Hence, we have arrived at the following formulation of the network reconfiguration problem:
$$ \min \{f^TRf \mid f \in \text{vert}(P)\}.$$

In particular, if all the resistances are uniform, then the problem is equivalent to finding a vertex with the smallest Euclidean norm, which is known to be NP-hard (see, for example, Lemma 4.1.4 in \cite{haddock_2018}). To the best of our knowledge, there do not exist any approximation algorithms for minimizing convex functions over vertices of a polytope. Even if we just require the solution to be integral (instead of lie at a vertex), and make strong assumptions on the objective function like strong convexity and Lipschitz gradients, Baes~\etal \cite{Baes_2012} have proved the following result:

\begin{theorem}[Theorem 2 in \cite{Baes_2012}]
Let $\mathcal{F} = P \cap \mathbb{Z}^n$ be presented by an oracle for solving quadratic minimization problems of the type $\min c^T x + \frac{\tau}{2} \|x\|_2^2$ with varying $c \in \mathbb{Q}^n$ and $\tau \in \mathbb{Q}_+$. There is no polynomial
time algorithm that can produce for every $\mathcal{F} = P \cap \mathbb{Z}^n$ and every strongly convex function $h: \mathbb{R}^n \to \mathbb{R}$ with Lipschitz gradients
a feasible point $\bar{x}$ such that $h(\bar{x}) - h (x^*) \leq n^2 - n$, where $x^* = \argmin_{x\in \mathcal{F}}h(x)$.
\end{theorem}

The authors also show that this $n^2 - n$ approximation is tight. Observe that this bound translates to an $m^2 - m$ approximation in the context of the network reconfiguration problem. However, the $\textsc{RIDe}$ algorithm we propose gives an $O(m-n)$ approximation in the stronger setting in which we require the solution to be a vertex of the flow polytope.

More recently, Hildebrand \etal \cite{Hildebrand_2016} show that there is an FPTAS for solving problems of the form $\min_{P \cap \mathbb{Z}^n} x^TQx$, where $P\subseteq \mathbb{R}^n$ is a polyhderon and $Q \in \mathbb{Z}^{n \times n}$ is a symmetric matrix with at most one negative eigenvalue. Recall that the objective function of the network reconfiguration is $f^T R f$, where $R$ is positive definite and symmetric. Hence, if we additionally assume all the resistances are integral, then using the result of Hildebrand \etal \cite{Hildebrand_2016}, there exists an FPTAS for minimizing the energy of the flow over integral flows. However, this result clearly does not extend to the network reconfiguration problem, since one can obtain an integral flow whose support does not form a spanning tree.

\section{Challenges for Generalizing Cut-based Results to Planar Graphs} \label{app:planar}
 {As mentioned in Section~\ref{sec:n_and_logn}, one potential generalization of the cut-based approximation results is to find such a family of cuts for planar graphs and to try to get an $\bigO(\sqrt{n})$-approximation for planar graphs with uniform edge resistances. In particular, the planar separator theorem can be interesting to solve this problem, as it guarantees the existence of small cuts for any planar graph. Here we discuss potential roadblocks towards this approach.}
 
 \noindent 
    {{\it 1. Edge separators:} All planar graphs do not have small edge separators. For example, if we consider a wheel graph with $n$ nodes as shown in Fig.~\ref{fig:limitations}-(left), any cut that splits the graph into two (almost) equal-size parts has at least a constant fraction of $n$ edges. Even though each planar graph has a vertex separator of size $\bigO(\sqrt{n})$, when we consider edge separators, one can only guarantee an edge separator of size $\bigO(\sqrt{\Delta n})$, where $\Delta$ is the maximum degree in the graph \cite{diks1988edge}}.
    
    \noindent 
    {{\it 2. Constrained separator cuts:} Suppose that the planar graph has a bounded degree, the second limitation is that the choice of the $k$th cut is constrained by the choice of previous cuts. The separator theorem is oblivious to the edge directions in the following sense: 
    once we find a cut of $\bigO(\sqrt{\Delta n})$ in our planar graph as shown in Fig.~\ref{fig:limitations}-(middle), it induces a natural direction on the corresponding edges, where the subsequent cuts have to respect those directions. In other words, if edge $(u_i,v_i)$ appears in our cut, there cannot be a future cut $S \ni r$ such that $v_i\in S$ and $u_i\notin S$. Since the opposite is valid ($u_i\in S, v_i\notin S$), contracting $u_i,v_i$ is not possible, without loss of generality.}
    \begin{figure}[t]
\centering
\begin{minipage}{.21\textwidth}
  \centering
  \includegraphics[width=3cm]{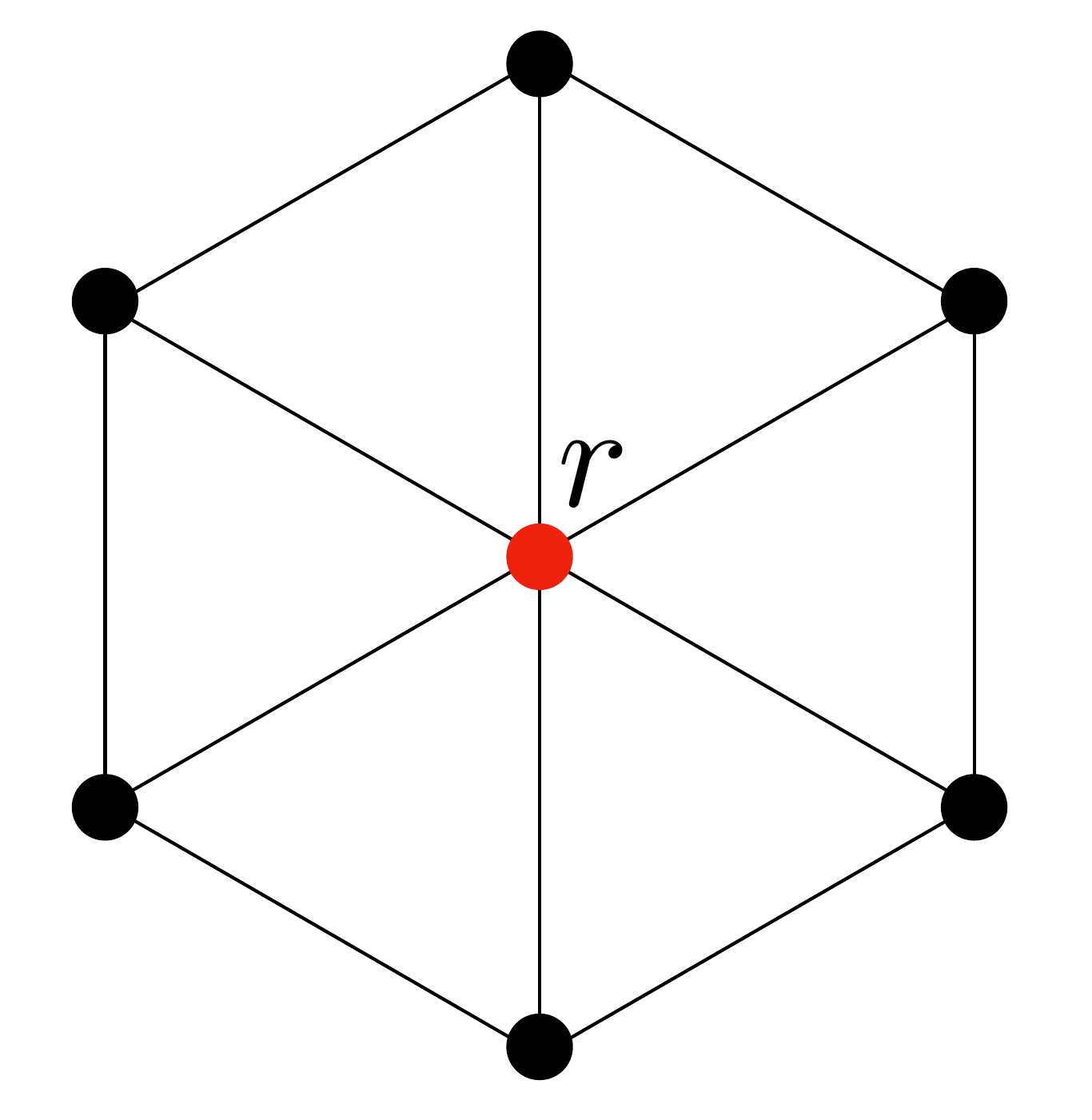}
\end{minipage}%
\begin{minipage}{.38\textwidth}
  \centering
  \includegraphics[width=5.5cm]{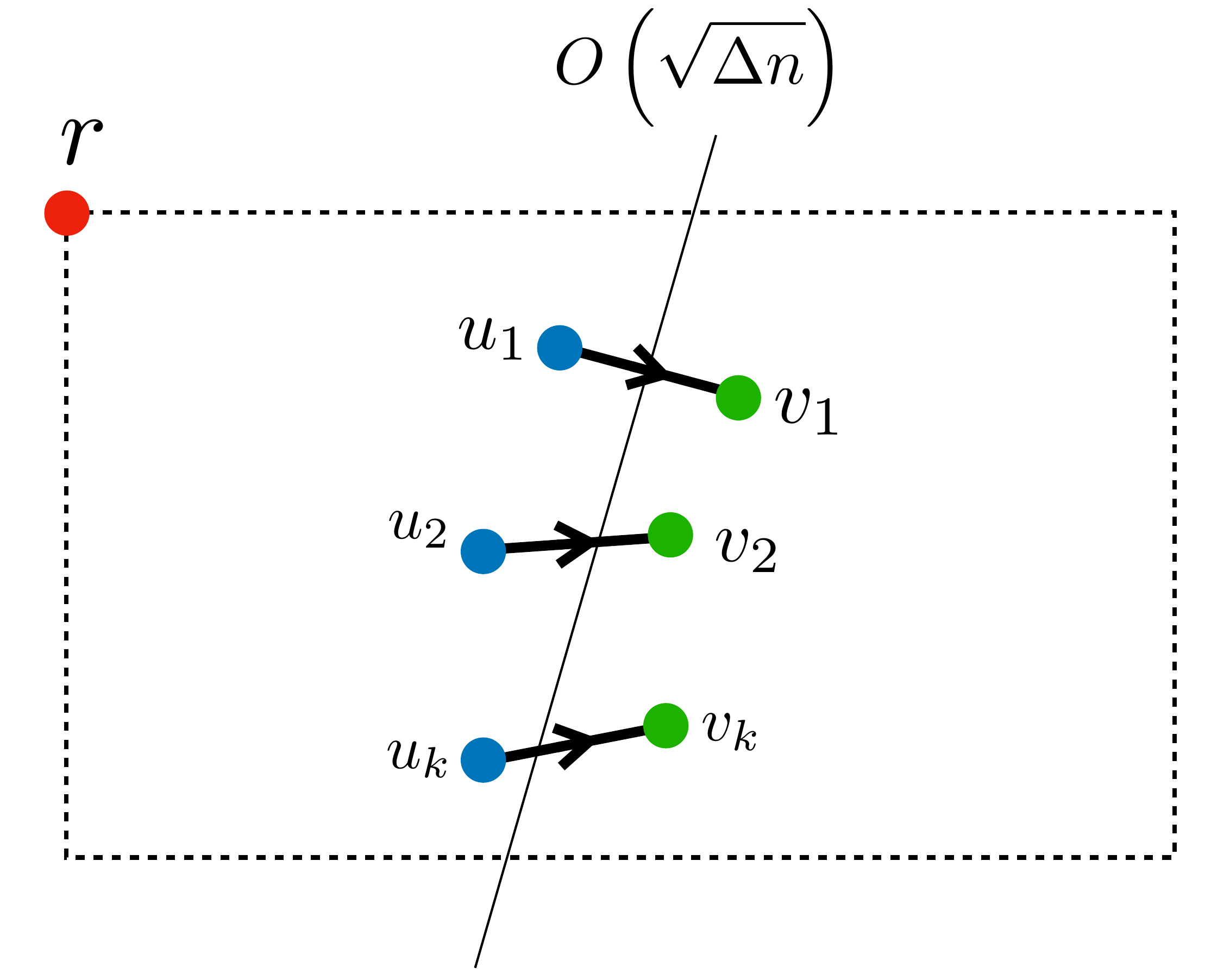}
\end{minipage}
\begin{minipage}{.38\textwidth}
  \centering
  \includegraphics[width=5.5cm]{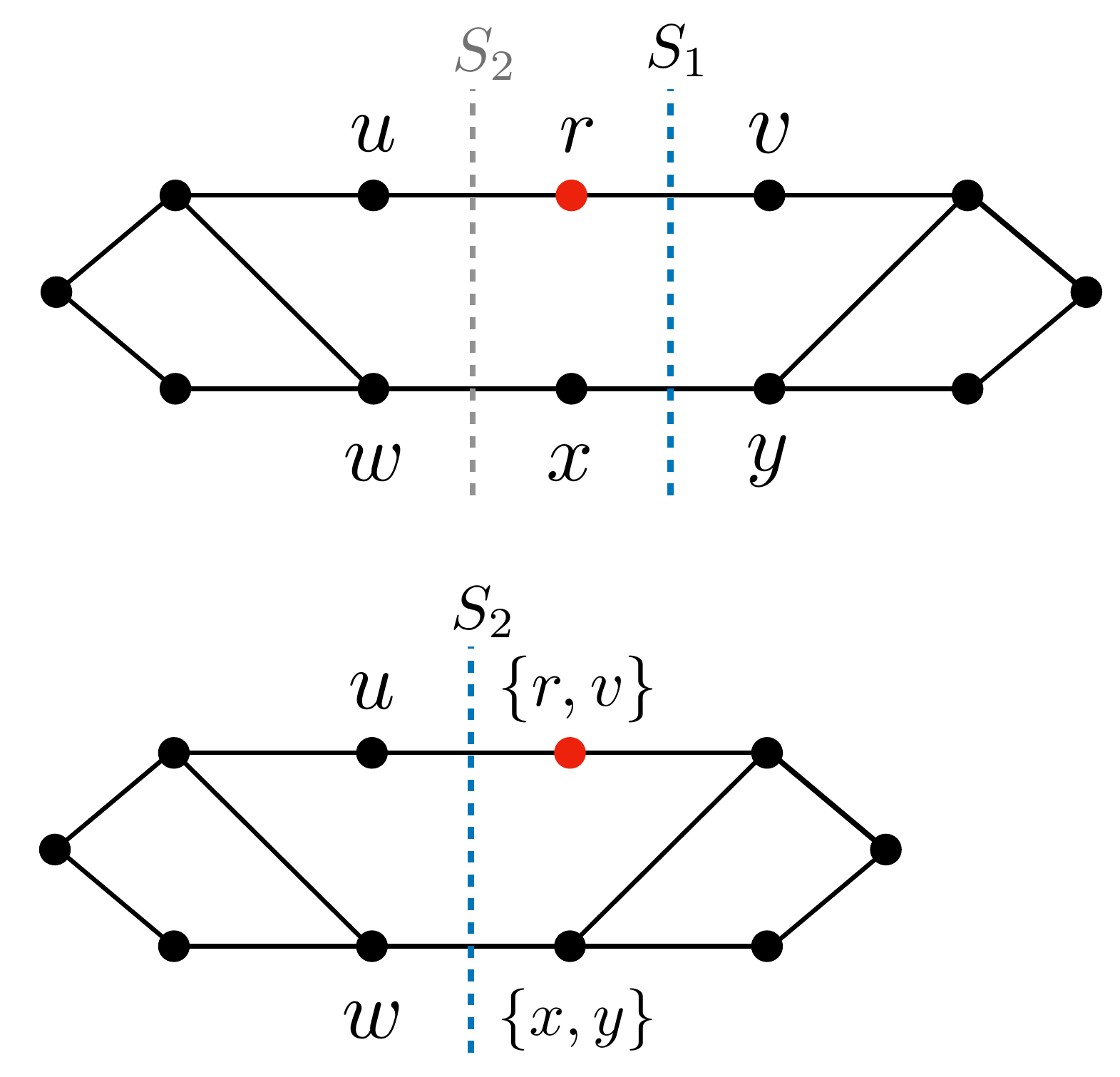}
\end{minipage}
\caption{Limitations of cut-based lower bounds for planar graphs.}
\label{fig:limitations}
\end{figure}
    
    {The example in Fig.~\ref{fig:limitations}-(right) shows how finding appropriate cuts might be challenging even with contracting edges (note that by considering edge contraction, we are already losing the option of cut intersections, hence searching over a smaller family of cuts). The idea in this figure is that after we find our first (small) cut $S_1$, we contract the edges that intersect the cut, in an attempt to avoid using those edges in violating directions in the future cuts. Now when we find our next cut $S_2$ in the contracted graph (shown at the bottom), notice that node $x$ becomes unreachable from the root. This is because cut $S_1$ induces the direction of edge $(x,y)$ to be from $x$ to $y$, and cut $S_2$ sets the direction of edge $(x,w)$ from $x$ to $w$. In other words, even with just 2 cuts, it has become infeasible to find a spanning tree that respects the cut directions.
    
   If one can somehow incorporate these one-way constraints in the separator oracle, i.e.,  once a node $v_i$ is picked in $S$, the corresponding $u_i$'s (could be more than one) from previous cuts should be picked in $S$ as well, then one could recursively use this oracle and hope for an $\bigO(\sqrt{n})$ approximation.}

\section{{Performance of Shortest-Path Trees}} \label{app:SPT}
\textcolor{blue}{The problem with shortest-path trees is that they are oblivious to the node demands. We now give a simple construction that demonstrates this behavior. Consider the complete undirected graph $G_n = (V,E)$ on $n$ nodes, where $n \geq 3$. Fix a root $r$ arbitrarily and let $P$ be any Hamiltonian path in $G_n$ starting at $r$. Now, suppose that $r_e = 1$ for all $e \in P$ and $r_e = n$ for all other edges $e \in E\setminus P$. Finally, let all the demands $d_i = 1$ for all $i \in V$.}
    
    \begin{figure}[t]
    \centering
    \includegraphics[scale = 0.4]{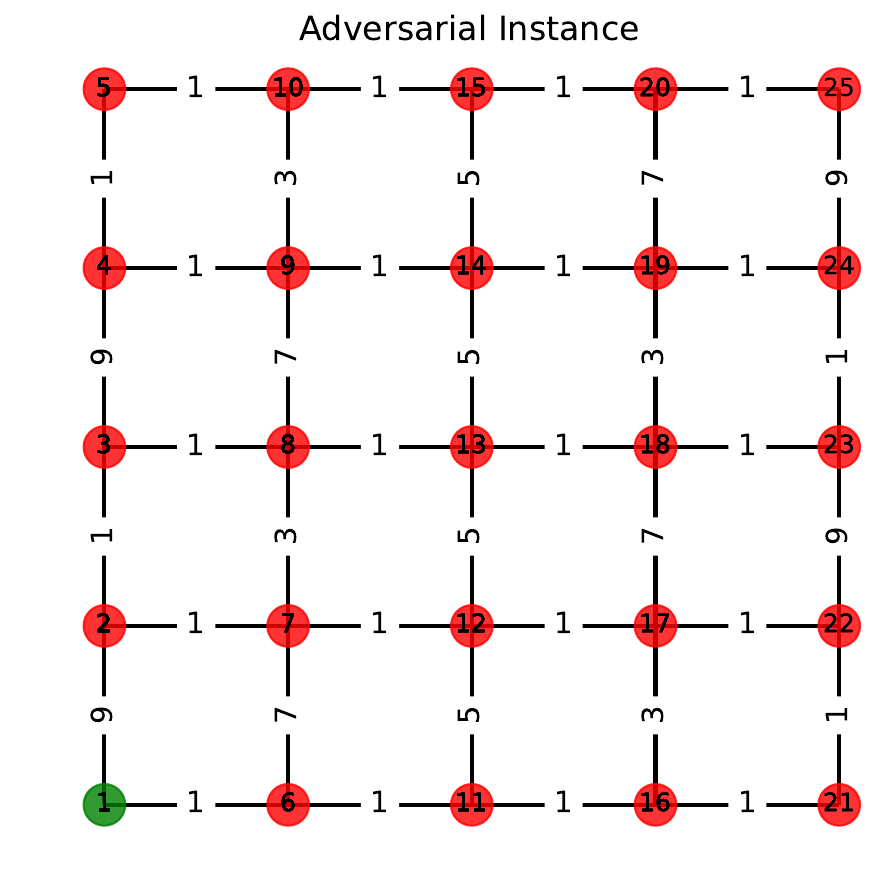}
    \includegraphics[scale = 0.4]{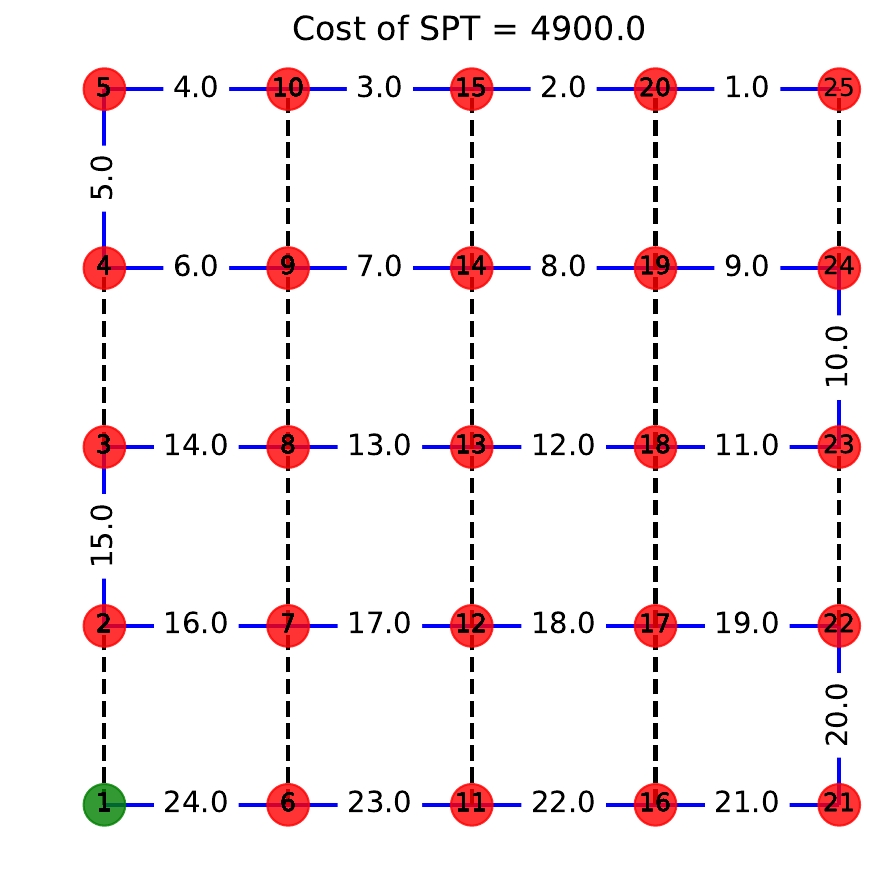}
    \centering
    \includegraphics[scale = 0.4]{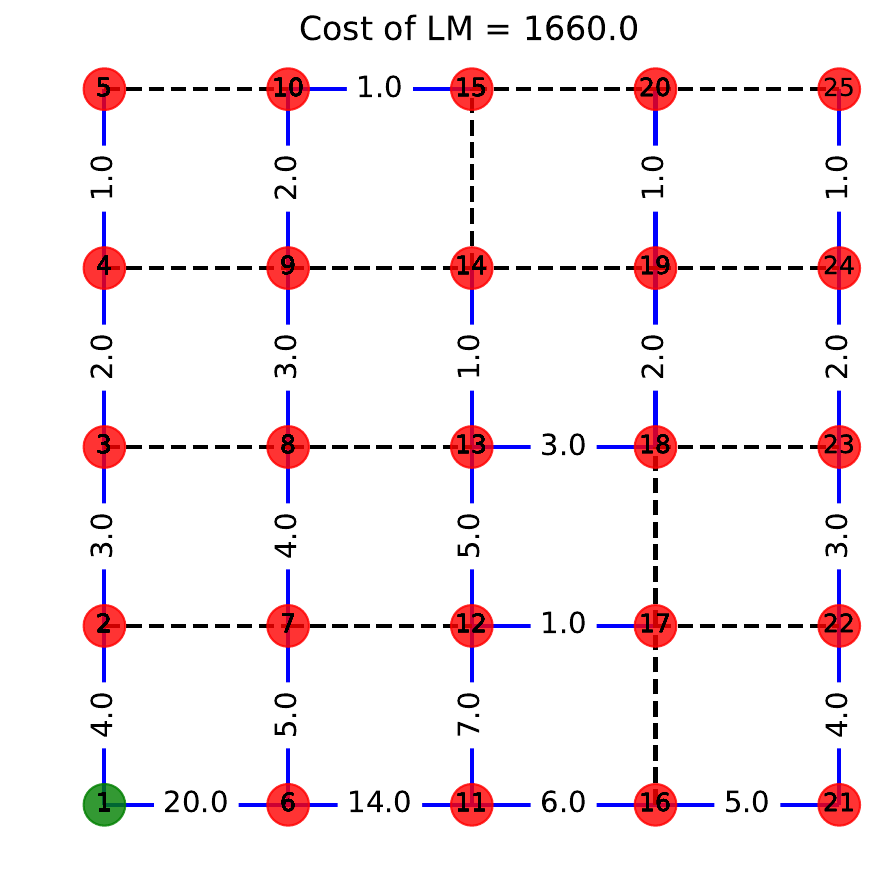}
    \includegraphics[scale = 0.4]{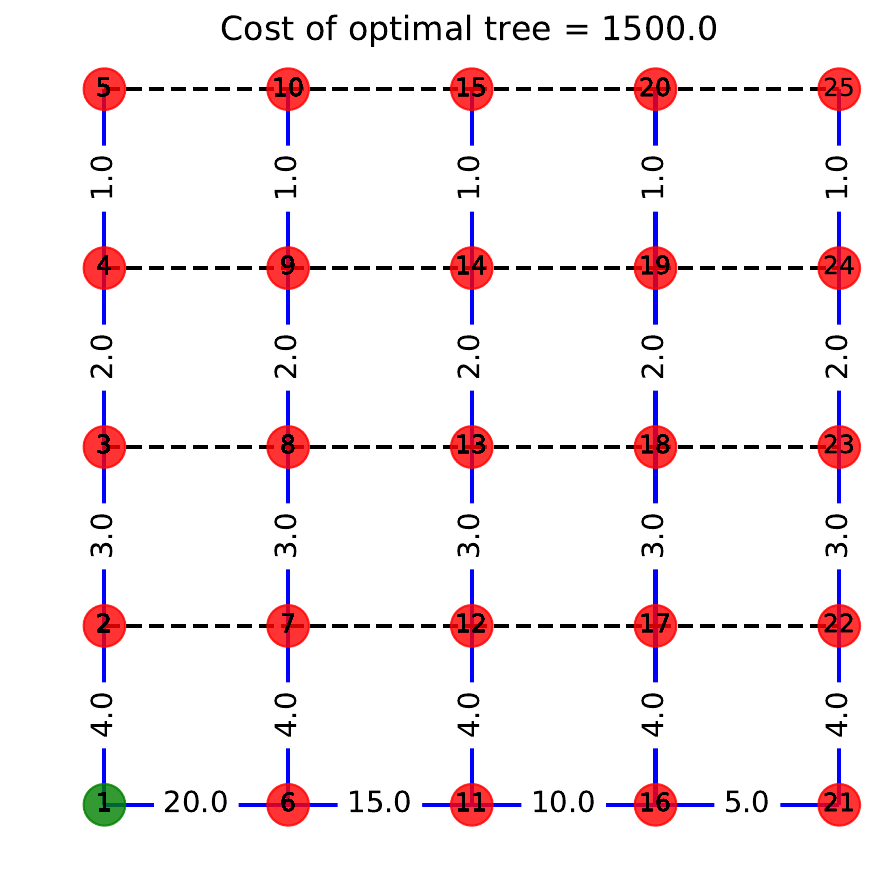}
    \caption{\small \textcolor{blue}{ \textbf{Top left:} An example of a $5\times 5$ grid, where the root $r$ is the green node in the bottom-left corner, all the demands are 1, and the resistances for each edge are given in the figure. These edge resistances are chosen in an adversarial manner to demonstrate how SPTs are oblivious to node demands. \textbf{Top right:} The SPT with respect to the adversarial edge resistances and its cost. \textbf{Bottom left:} The LM tree for adversarial instance and its cost. \textbf{Bottom right:} The optimal tree for adversarial instance and its cost}.}
    \label{fig:adver}
\end{figure}
    
\textcolor{blue}{By construction, in this example the shortest path tree rooted at $r$ will be $P$. The cost of that tree is 
    $$\sum_{i=1}^{n-1} (1)i^2 = \frac{1}{6} (n-1)n(2n - 1) = \frac{1}{6}(2n^3 - 3n^2 + n) \geq \frac{1}{6}(2n^3 - 3n^2)  \geq \frac{n^3}{6},$$
    where the last inequality follows since $n \geq 3$.
    We now compare this with the LM Heuristic, which also attains the optimal tree in this example. The LM Heuristic will start by computing a BFS tree rooted at $r$. Since $G_n$ is a complete graph, this tree will only have two layers: (i) layer 1 will just contain $r$, (ii) layer 2 will contain all other nodes in the graph. Now, the LM Heuristic will match the nodes in layer 2 to the nodes in layer 1 by solving the LP given in the main paper. However in this specific case, there is only one possible matching: match all nodes in layer 2 to the root node $r$ in layer 1. This gives us a final tree, where all nodes $i \in V\setminus \{r\}$ are connected to the root through the edge $(r,i)$. The cost of that tree is $(1)(1)^2 + (n)(n-1)(1)^2 = n^2 - n + 1 \leq n^2$. Hence the gap between the SPT and the LM tree (which is also the optimal tree) is at least $n/6$ or $\Omega(n)$. Note that the same gap could be achieved by the wheel gap, which is planar.}

\textcolor{blue}{We can also extend this construction to grid graphs. Consider and $n \times n$ grid and fix a root $r$ in the corner of the grid. Further, consider a DFS tree rooted at $r$. A DFS tree in this case will be a Hamiltonian path $P$ with $n^2-1$ edges that starts at $r$ and visits every other node in the graph. Now, suppose that $r_e = 1$ for all $e \in P$ and the resistances for other edges are greedily chosen as small as possible so that the shortest-paths tree coincides with $P$. Finally, let all the demands $d_i = 1$ for all $i \in V$. We give an example for a $5\times 5$ grid in Figure \ref{fig:adver}}.

\begin{figure}[t]
\centering
\includegraphics[scale = 0.5]{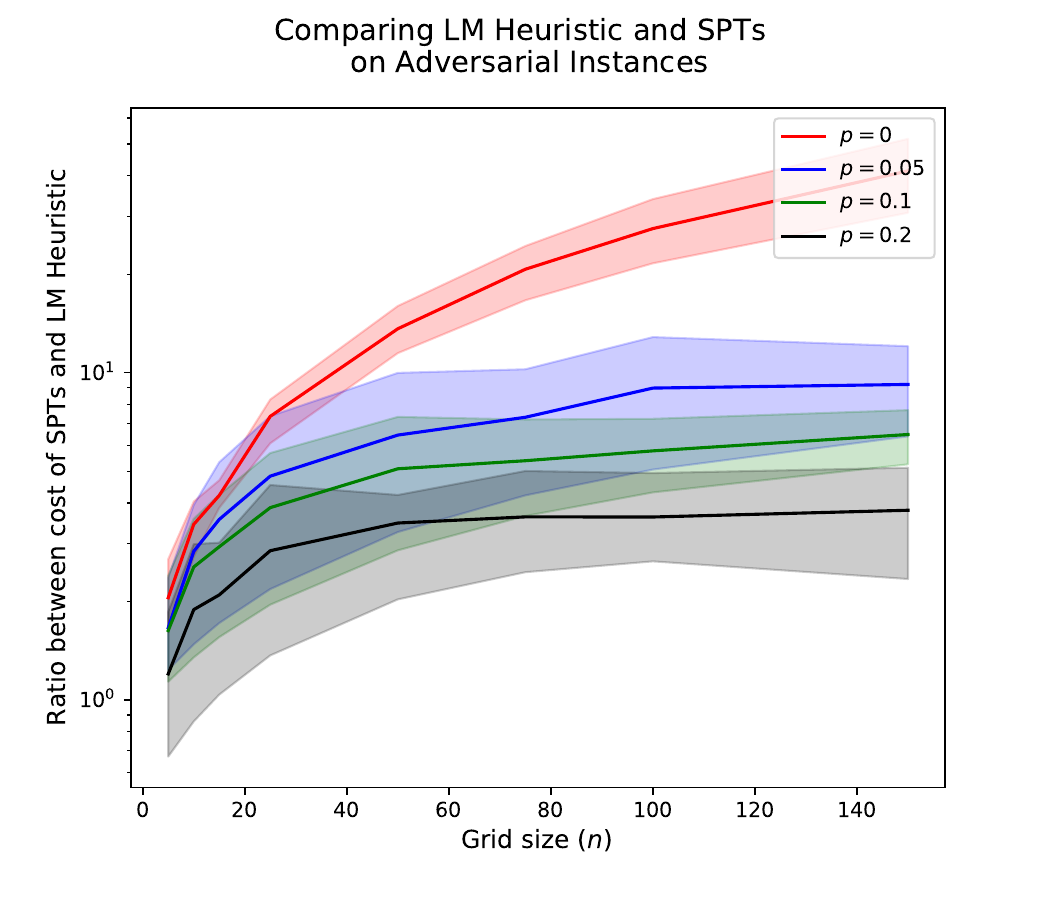}
\caption{\small \textcolor{blue}{Comparison of LM Heuristic and SPTs on adversarial instances for varying $n \times n$ grid sizes (the $x$-axis represents the value of $n$). In those instances, we also add noise to the edge resistances and consider sparsification probabilities $p \in \{0, 0.05,0.1,0.2\}$.}}
\label{fig:adver2}
\end{figure}

\textcolor{blue}{We see in Figure \ref{fig:adver} that the LM tree is near optimal and the ratio between the cost of SPT and LM tree is around 3. We extend this construction for varying grid sizes, and in line with the original experiments presented in the main body of the paper, we add noise to the edge resistances that is sampled from a normal distribution with mean $0$ and standard deviation of 0.5, while also considering sparsification probabilities $p\in \{0, 0.05,0.1,0.2\}$. This is how we generated the resistances for the adversarial second set of instances used in the computations. In Figure \ref{fig:comp}-(right), we present the performance of our algorithms on those adversarial instances, where we find that LM is the best performing algorithm. The best feasible solution was obtained by running the MIP with the LM output as a warm-start in 64\% of the instances, and Branch Exchange initialized with the LM output in 36\% of the instances.}

\textcolor{blue}{As shown in Figure \ref{fig:adver2}, we find that the cost of SPTs can be arbitrarily bad due to the fact that they are oblivious to node demands. We also see that as the sparsification probability increases, the ratio between SPTs and LM decreases, which is expected as the cost of the DFS tree decreases as the sparsification probability increases.}

\section{Computations} \label{computaions-appendix}

\subsection{MIP Formulation Used in Computations}
Let $G=(V,E)$ ($\abs{V}=n$, $\abs{E}=m$) be a connected and undirected graph with: root $r \in V$, resistances $r_e>0$ for each edge $e\in E$ and demands $d_i\geq 0$ for each node $i\in V\backslash \{r\}$ supplied by the root node (which implies that $d_r = - \sum_{i\in V\setminus \{r\}} d_i$). Also, let $\delta^+(v)$ and $\delta^-(v)$ denote the sets of incoming and outgoing edges of $v$ (after fixing an arbitrary orientation on the edges).
To incorporate acyclic support constraints we utilized Martin's \cite{martin1991} extended formulation for spanning trees, which has $O(n^3)$ constraints and variables. The {\em network reconfiguration} problem could be formulated as follows:
\begin{equation*} \label{formulation2}
\begin{aligned}
     \min &~~ \sum_{e \in E} r_e f_e^2\\
    \text{s.t}     &~~ \sum_{e \in \delta^+(u)} f_e - \sum_{e\in \delta^-(u)} f_e = d_u &&  \forall \; u \in V\\
    &~~ \sum_{e \in E} x_e = n-1 \\
     &~~x_{\{v,w\}} = z_{v,w,u} + z_{w,v,u} &&  \forall\; \{v,w\} \in E, u \in V\setminus \{v,w\}\\
   &~~ \displaystyle x_{\{v,w\}} + \sum_{\substack{u \in V\setminus \{v,w\}: \\ \{v,u\} \in E}} z_{v,u,w} = 1 &&  \forall\; \{v,w\} \in E\\
    &~~ -M x_e \leq f_e \leq Mx_e  &&  \forall\; e \in E\\
          &~~x_{\{v,w\}}  \in \{0,1\}, z_{v,w,u}  \in \{0,1\}  &&  \forall \; \{v,w\} \in E, u \in V \setminus \{v,w\}
\end{aligned}
\end{equation*}
where $M$ is a sufficiently large scalar. In particular it suffices to set $M =\sum_{u \in V\setminus \{r\}} d_u = -d_r$. Note that the pairs $\{v, w\}$ are unordered, while $v, v, w$ are an ordered triple of distinct vertices with $\{v, w\} \in E$ (and $u \in V\setminus \{v,w\}$).
 
 \subsection{More computational Plots}
Recall that for our first set of computational experiments, we constructed instances on $25 \times 25$ grids with sparsification probability $p \in \{0.05,0.1, 0.2\}$. Moreover we consider demands and resistances randomly chosen in $[0.5,1.5]$ and $[1,10]$ respectively. We benchmarked the performance of depth-first search (DFS) trees, \textsc{RIDe}, the Layered-Matching (LM) heuristic, the Branch Exchange heuristic and the convex integer program given in the previous section. \textcolor{blue}{In the main body of the paper we only presented the plots for when the sparsification probability was $p = 0.2$. Here, we present additional plots (with more details about the performance of \textsc{Ride} and Branch Exchange initialized with the LM tree) for varying sparsification probabilities $p \in \{0.05,0.1, 0.2\}$}.

\begin{figure}[H]
    \centering
        \includegraphics[scale = 0.45]{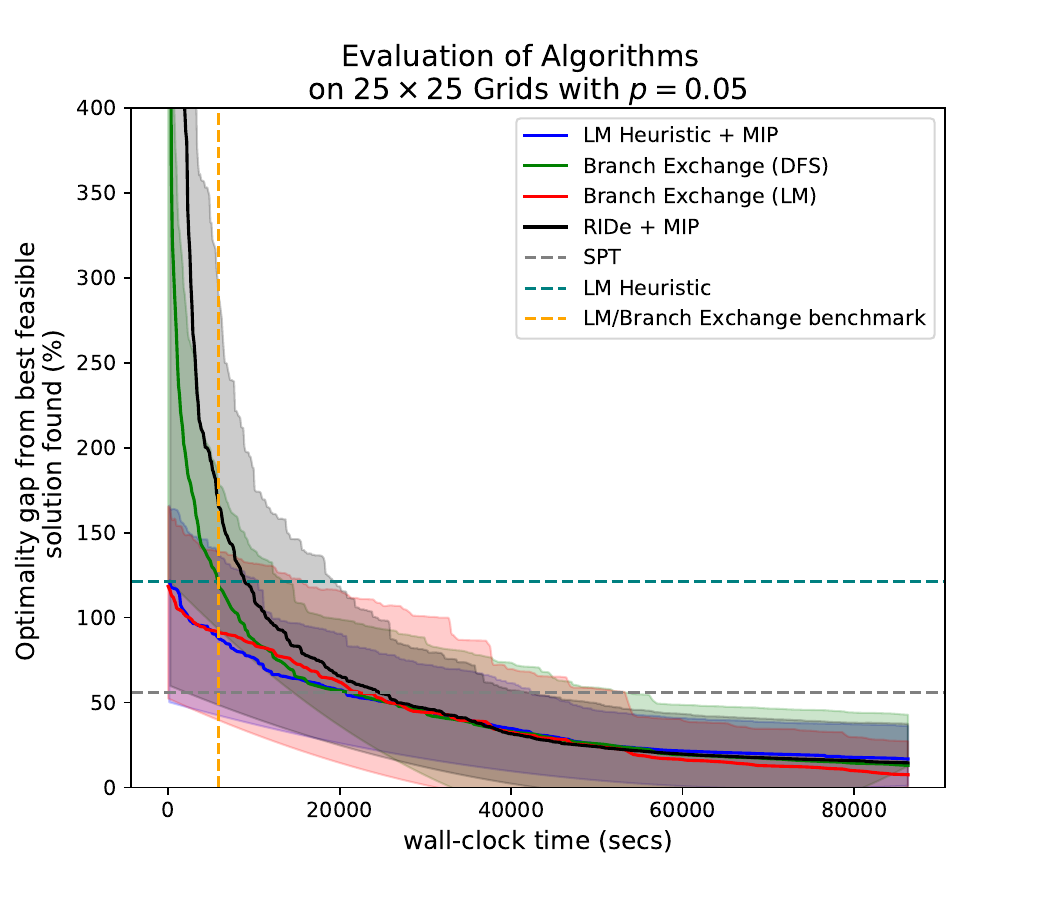}
        \includegraphics[scale = 0.45]{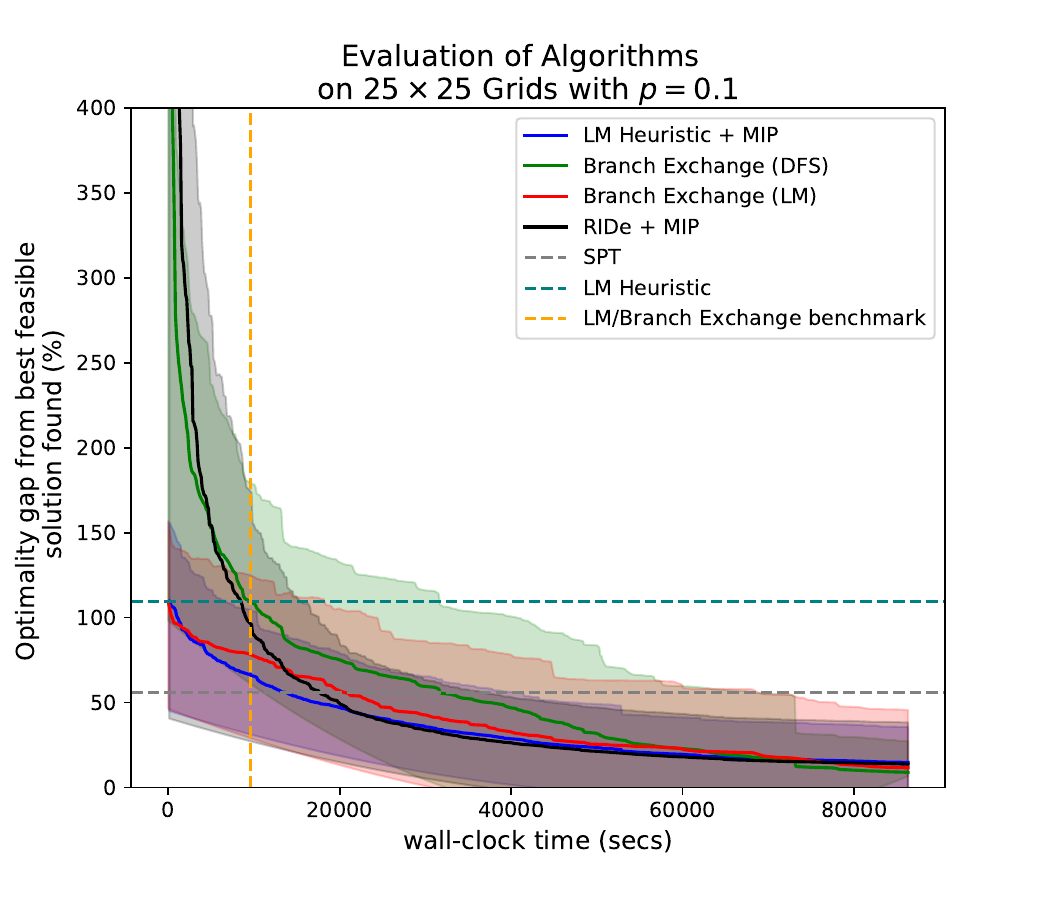}
        \includegraphics[scale = 0.45]{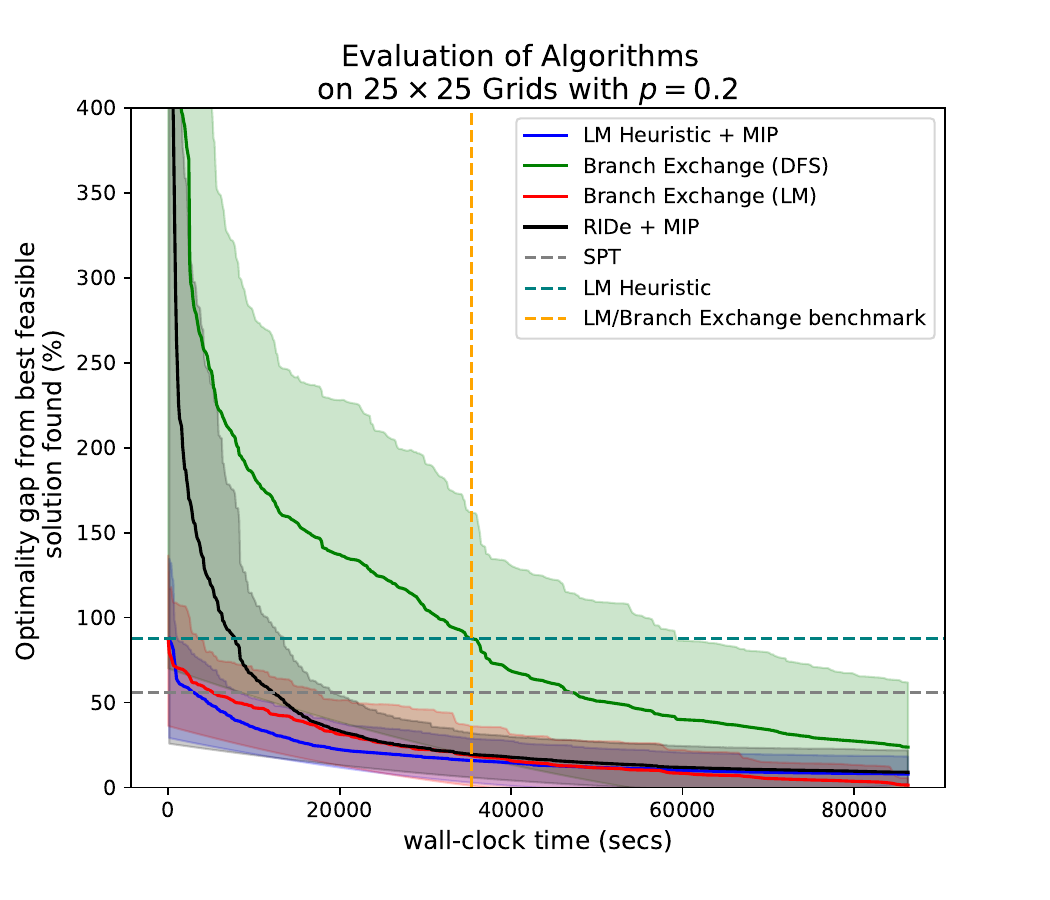}
    \caption{Plot comparing performance of different algorithms and MIP with different warm-start solutions on sparse $25 \times 25$ with sparsifcation probability $p \in \{0.05,0.1,0.2\}$. The demands and resistances are randomly chosen in $[0.5,1.5]$ and $[1,10]$ respectively. The dotted horizontal line compares the quality of the Layered Matching heuristic solution and the color margins represent confidence intervals across the different instances.}
    \label{fig:comp1}
    \vspace{-10 pt}
\end{figure}

\end{document}